\documentclass[11pt,a4paper]{amsart}
\usepackage{fullpage}

\usepackage[T1]{fontenc}
\usepackage[utf8]{inputenc}
\usepackage[english]{babel}

\usepackage{enumerate}

\usepackage{listings}
\usepackage[dvipsnames]{xcolor}
\usepackage{eqnarray, amstext}
\usepackage{qsymbols,xparse}
\usepackage{algpseudocode}
\usepackage[Algorithm,ruled]{algorithm}
\usepackage{nicefrac}

\newcommand\eg{\emph{e.g.}\ }

\usepackage{syntax}

\usepackage{caption}
\usepackage{needspace}

\theoremstyle{definition}
\newtheorem{theorem}{Theorem}

\newtheorem{lemma}[theorem]{Lemma}
\newtheorem{prop}[theorem]{Proposition}

\newtheorem{cor}[theorem]{Corollary}
\newtheorem{example}[theorem]{Example}

\newcommand{\KExpansions}[1]{\textsc{K-Expansions}(#1)}
\newcommand{\SExpansions}[1]{\textsc{S-Expansions}(#1)}
\newcommand{\RewritingSet}[2]{\textsc{RewritingSet}(#1,#2)}
\newcommand{\MeshSet}[2]{\textsc{MeshSet}(#1,#2)}
\newcommand{\potential}[1]{\pi(#1)}
\newcommand{\rank}[1]{\rho(#1)}
\newcommand{\gammaF}[1]{\Gamma(#1,1)}

\begin{document}

\title{Normal-order reduction grammars}
\author{Maciej Bendkowski}
\address{
  Theoretical Computer Science Department\\
  Faculty of Mathematics and Computer Science\\
  Jagiellonian University\\
  ul. Prof. {\L}ojasiewicza 6, 30-348 Krak\'ow, Poland}
\email{bendkowski@tcs.uj.edu.pl}
\thanks{This work was partially supported within the grant 2013/11/B/ST6/00975
            founded by the Polish National Science Center.} 

\maketitle

\begin{abstract}
We present an algorithm which, for given $n$, generates an unambiguous regular tree grammar defining the set of combinatory logic terms, over the set $\{S,K\}$ of primitive combinators, requiring exactly $n$ normal-order reduction steps to normalize. As a consequence of Curry and Feys's standardization theorem, our reduction grammars form a complete syntactic characterization of normalizing combinatory logic terms. Using them, we provide a recursive method of constructing ordinary generating functions counting the number of $S K$-combinators reducing in $n$ normal-order reduction steps. Finally, we investigate the size of generated grammars, giving a primitive recursive upper bound.
\end{abstract}

\section{Introduction}

Since the time of the pioneering works of Moses Schönfinkel~\cite{schonfinkel} and Haskell Curry~\cite{curry}, combinatory logic is known as a powerful, yet extremely simple in structure, formalism expressing the notion of computability. With the dawn of functional programming languages in the early 1970s,
combinatory logic, with its standard normal-order reduction scheme~\cite{curry-feys1958}, is used as a practical implementation of
lazy semantics in languages such as SASL~\cite{turner1979} or its successor
Miranda~\cite{turner1986}. Lack of bound variables in the language resolves
the intrinsic problem of substitution in \mbox{$\lambda$-calculus}, making the reduction relation a simple computational step and so, in consequence, the leading workhorse in implementing call-by-need reduction schemes.

Surprisingly, little is known about the combinatorial properties of normal-order reduction and, in particular, its behaviour in the `typical' case of large random combinators. With the growing popularity of random software testing (see, \eg \cite{DBLP:conf/icse/PalkaCRH11}) `typical' properties of random $\lambda$-terms and combinators became of immense practical importance. In this approach to software verification, large random terms are generated and used to check the programmer-declared function invariants, making it crucial to understand and exploit the semantic properties of so generated terms. 

State-of-the-art research in this field includes counting and generating $\lambda$-terms (see \eg~\cite{DBLP:journals/jfp/GrygielL13}~\cite{DBLP:journals/corr/Lescanne14}~\cite{DBLP:conf/stacs/GittenbergerG16}), their restricted classes~\cite{1510.01167}, investigating their asymptotic properties~\cite{lmcs:848}~\cite{Bendkowski2016} as well as the asymptotic properties of combinatory logic~\cite{Bendkowski2015}. 

Main tools used in this line of research include formal power series and generating functions. Interested in a particular counting sequence ${(a_n)}_{n \in \mathbb{N}}$ corresponding to a set of terms $A$, we construct a suitable generating function, which treated as a complex function in one variable $z$ yields a Taylor series expansion around $z=0$ with coefficients forming our sequence ${(a_n)}_{n \in \mathbb{N}}$. Methods of analytic combinatorics~\cite{Flajolet:2009:AC:1506267} allow us to derive, sometimes surprisingly accurate, asymptotic approximations of the growth rate of ${(a_n)}_{n \in \mathbb{N}}$ and, in consequence, use them to study the asymptotic behaviour of $A$. 

Finding appropriate generating functions plays therefore an important role in the process of investigating properties of `typical' terms. In~\cite{Bendkowski2015}, authors investigated the asymptotic density of weakly normalizing terms in the set of all combinators, showing that a `typical' combinator cannot have a trivial $0$ -- $1$ asymptotic probability of normalization. The result was obtained by constructing large classes of terms with and without the normalization property. Though sufficient for the purpose of showing the non-trivial behaviour of normalization, their classes reveal the combinatorial structure of just a small asymptotic portion of normalizing terms.

In this paper we give a complete combinatorial characterization of normalizing combinatory logic terms over the set $\{S, K\}$ of primitive combinators. We construct a recursive family ${\{R_n\}}_{n \in \mathbb{N}}$ of regular tree grammars defining combinators reducing in exactly $n$ normal-order reductions. By Curry and Feys's standardization theorem~\cite{curry-feys1958}, normal-order evaluation of normalizing combinators leads to their normal forms, hence our \emph{normal-order reduction grammars} form a complete partition of normalizing combinators. Our approach is algorithmic in nature and provides fully automated methods for constructing ${\{R_n\}}_{n \in \mathbb{N}}$ as well as their corresponding ordinary generating functions.

The paper is organized as follows. In Sections~\ref{sec:combinatory-logic}, and~\ref{sec:regular-tree-grammars} we give preliminary definitions and notational conventions. In Section~\ref{sec-pseudocodes} we explain our pseudo-code notation and related implementation. In Section~\ref{sec:algorithm} we present a high-level overview on the algorithm. In Section~\ref{sec:analysis} we analyse the algorithm giving proofs of soundness~\ref{sec:soundness}, completeness~\ref{sec:completeness} and unambiguity~\ref{sec:unambiguity}. In Section~\ref{sec:gen-funs} we give a recursive construction of ordinary generating functions corresponding to ${\{R_n\}}_{n \in \mathbb{N}}$. In Section~\ref{sec:applications} we discuss some consequences and applications of normal-order reduction grammars. Finally, in Section~\ref{sec:upper-bound} we investigate the size of the generated grammars.

\subsection{Combinatory Logic}\label{sec:combinatory-logic}
We consider the set of terms over primitive combinators $S$ and $K$. In other words, the set $\mathcal{C}$ of combinatory logic terms defined as $\mathcal{C} := S~|~K~|~\mathcal{C}\, \mathcal{C}$. We follow standard notational conventions (see e.g.~\cite{barendregt1984}) --- we omit outermost parentheses and drop parentheses from left-associated terms, e.g.~instead of $((S K) (K K))$ we write $S K (K K)$. We use $\to_w$ to denote the normal-order reduction relation (reduce the leftmost outermost redex) to which we usually refer briefly as the reduction relation. We use lower case letters $x,y,z,\ldots$ to denote combinatory logic terms. For an introduction to combinatory logic we refer the
reader to~\cite{barendregt1984},~\cite{curry-feys1958}.

\subsection{Regular tree grammars}\label{sec:regular-tree-grammars}
In order to characterize terms normalizing in $n$ steps we use
regular tree grammars (see e.g.~\cite{tata2007}), a generalization of regular
word grammars. A \emph{regular tree grammar}
$G = (S, N, \mathcal{F}, P)$ consists of an axiom $S$, a set $N$ of non-terminal
symbols such that $S \in N$, a set of terminal symbols $\mathcal{F}$ with
corresponding arities and a finite set of production rules $P$ of the form
$\alpha \to \beta$ where $\alpha \in N$ is a non-terminal and $\beta \in
T_{\mathcal{F}}(N)$ is a term in the corresponding term algebra
$T_{\mathcal{F}}(N)$, i.e.~the set of directed trees built upon terminals
$\mathcal{F}$ according to their associated arities. To build terms of grammar
$G$, we start with the axiom $S$ and use the corresponding derivation relation,
denoted by $\to$, as defined through the set of production rules $P$.

\begin{example} 
Consider the following regular tree grammar defined as $B = (S,N,\mathcal{F},P)$ where
$S := \mathcal{B}$, $N := \{\mathcal{B}\}$, $\mathcal{F} := \{`(!), `(?)(\cdot,\,
\cdot)\}$, and $P$ consists of the two following rules: 
\begin{displaymath}
\left\{\begin{array}{r@{}l@{\qquad}l}
    \mathcal{B} &\to `(?)(\mathcal{B}, \mathcal{B})\\
    \mathcal{B} &\to `(!)
  \end{array}\right.
  \end{displaymath}
Note that $B$ defines the set of terms isomorphic to plane binary trees where
leafs correspond to the nullary constant $`(!)$ and inner nodes correspond to
the binary terminal $`(?)(\cdot,\, \cdot)$.
\end{example}

In our endeavour, we are going to recursively construct regular tree grammars
generating sets of combinatory logic terms. We set a priori their axioms and
both terminal and non-terminal symbols, leaving the algorithm to define the
remaining production rules. And so, the $n$th grammar $R_n$ will have:
\begin{enumerate}[(i)]
    \item an axiom $S = R_n$,
    \item a set $\mathcal{F}$ of terminal symbols consisting of two nullary
            constants $S$, $K$ and a single binary application operator,
    \item a set of non-terminal symbols $N = \{\mathcal{C}\} \cup
            \{R_0,\ldots,R_n\}$ where $\mathcal{C}$ denotes the axiom of the set
            of all combinatory logic terms, as defined in the previous section.
\end{enumerate}
In other words, the grammar $R_n$ defining terms normalizing in $n$ steps, will
reference all previous grammars $R_0,\ldots,R_{n-1}$ and the set of all
combinatory logic terms $\mathcal{C}$. 

Throughout the paper, we adopt the following common definitions and notational
conventions regarding trees. We use lower case letters
$\alpha,\beta,\gamma,\delta,\ldots$ to denote trees, i.e.~elements of the term
algebra $T_{\mathcal{F}}(N)$ where $N = \{\mathcal{C}\} \cup \{R_0,\ldots,R_n\}$
for some $n$. Whenever we want to use a combinator without specifying its type,
we use capital letters $X,Y,\ldots$. We define the \emph{size} of $\alpha$ as
the number of applications in $\alpha$.  We say that $\alpha$ is \emph{normal}
if either $\alpha$ is of size $0$, or $\alpha = X \alpha_1 \ldots \alpha_m$, for
some $m \geq 1$, where all $\alpha_1,\ldots,\alpha_m$ are normal. In the latter
case we say moreover that $\alpha$ is \emph{complex}. Since we are going to work exclusively with normal trees, we assume that all trees are henceforth normal.
We say that a complex $\alpha$ is of \emph{length} $m$ if $\alpha$ is in form of $X \alpha_1 \ldots \alpha_m$. Otherwise, if $\alpha$ is not complex, we say that it is of length $0$. The \emph{degree} of $\alpha$, denoted as $\rank{\alpha}$, is the minimum natural number $n$ such that $\alpha$ does not contain references to any $R_i$ for $i \geq n$. In particular, if $\alpha$ does not reference any reduction grammar, its degree is equal to $0$. We use $L_G(\alpha)$ to denote the language of $\alpha$ in grammar $G$. Since $R_n$ does not reference grammars of greater index, we have
$L_{R_{\rank{\alpha}-1}}(\alpha) = L_{R_n}(\alpha)$ for arbitrary $n \geq \rank{\alpha}$.
And so, for convenience, we use $L(\alpha)$ to denote the language of $\alpha$ in grammar $R_{\rank{\alpha}-1}$ if $\rank{\alpha} > 0$. Otherwise, if 
$\rank{\alpha} = 0$ we assume that $L(\alpha)$ denotes the language of $\alpha$ in grammar $\mathcal{C}$. Finally, we say that two normal trees are
\emph{similar} if both start with the same combinator $X$ and are of equal
length.

\begin{example}
        Consider the following trees:
        \begin{enumerate}[(i)]
                \item $\alpha = S (K R_1) \mathcal{C}$, and
            \item $\beta = K (\mathcal{C} S) R_0$.
        \end{enumerate}
        Note that both $\alpha$ and $\beta$ are of size $3$ and of equal length
        $2$, although they are not similar since both start with
        different combinators. Moreover, only $\alpha$ is normal as $\beta$ has
        a subtree $\mathcal{C} S$, which is of positive size, but does not start
        with a combinator. Since $\alpha$ contains a reference to $R_1$ and no
        other reduction grammar, its degree is equal to $2$, whereas the degree of
        $\beta$ is equal to $1$.
\end{example}

A crucial observation, which we are going to exploit in our construction, is the fact that normal trees \emph{preserve} length of generated terms. In other words, if $\alpha$ is of length $m \geq 1$, then any term $x \in L(\alpha)$ is of length $m$ as well, i.e.~$x = X x_1 \ldots x_m$.

\subsection{Pseudo-codes and implementation}\label{sec-pseudocodes}

We state our algorithm using functional pseudo-codes formalising key design subroutines. The adopted syntax echoes basic Haskell notation and build-in primitives, though we use certain abbreviations making the overall presentation more comprehensible. And so, we use the following data structure representing normal trees.

\begin{lstlisting}
    -- | Normal trees.
    data Tree = S | K | C | R Int
              | App Tree Tree
\end{lstlisting}
 
In our subroutines, we use the following `syntactic sugar' abbreviating the structure of normal trees.

\begin{lstlisting}
    -- | Syntactic sugar.
    X a_1 ... a_m := App X (App a_1 (... App a_{m-1} a_m) ...)
\end{lstlisting}

Moreover, we allow the use of this abbreviated notation in pattern matching, meaning that by writing \verb|(X a_1 ... a_m)| we expect a complex tree of length $m$ for some $m \in \mathbb{N}$. If multiple arguments are supposed to share the same length, we use the same natural number $m$, e.g.~\verb|(X a_1 ... a_m)| and \verb|(X b_1 ... b_m)|.
A working Haskell implementation of our algorithm is available at~\cite{mb-haskell-implementation}.

\section{Algorithm}\label{sec:algorithm}

The key idea used in the construction of reduction grammars is to generate new productions in $R_{n+1}$ based on the productions in $R_n$. Necessarily, any term normalizing in $n+1$ steps reduces directly to a term normalizing
in $n$ steps, hence their syntactic structure should be closely related. As the base of our inductive construction, we use the set of normal forms $R_0$ given by
 \begin{equation*}
 R_0 := S~|~K~|~S R_0~|~K R_0~|~S R_0 R_0.
 \end{equation*}
 
 Clearly, primitive combinators $S$ and $K$ are in normal form. If we take a normal form $x$, then both $S\, x$ and $K\, x$ are again normal since we did not create any new redex. For the same reason, any term $S x_1 x_2$ where $x_1$ and $x_2$ are normal forms, is itself in normal form. And so, with the above grammar we have captured exactly all redex-free terms.

Let us consider productions of $R_0$. Note that from both the cases of $S R_0$ and $K R_0$ we can abstract a more general rule --- if $x$ reduces in $n$ steps, then $S x$ and $K x$ reduce in $n$ steps as well, since after reducing $x$ we have no additional redexes left to consider. It follows that any $R_n$ should contain productions $S R_n$ and $K R_n$. Similarly, from the case of $S R_0 R_0$ we can abstract a more general rule --- if $S x_1 x_2$ reduces in $n$ steps, then both $x_1$ and $x_2$ must reduce in total of $n$ steps. The normal-order reduction of $S x_1 x_2$ proceeds to normalize $x_1$ and $x_2$ sequentially. As there is no head redex, after $n$ steps we obtain a term in normal form. And so, $R_n$ should also contain productions $S R_i R_{n-i}$ for $i \in \{0,\ldots,n\}$.
 
 As we have noticed, all the above productions do not contain head redexes and hence do not increase the total amount of required reduction steps to normalize. Formalizing the above observations, we say that $\alpha$ is \emph{short} if either $\alpha = X \alpha_1$ or $\alpha = S \alpha_1 \alpha_2$. Otherwise, $\alpha$ is said to be \emph{long}. Hence, we can set a priori the short productions of $R_n$ for $n \geq 1$ and continue to construct the remaining long productions. Naturally, as we consider terms over two primitive combinators $S$ and $K$, we distinguish two types of long productions, i.e.~\textsc{S-} and \textsc{K-Expansions}.
 
 \subsection{K-Expansions}\label{sec:k-expansions}
Let us consider a production $\alpha = X \alpha_1 \ldots \alpha_m$ where $m \geq 0$. The set $\KExpansions{\alpha}$ is defined as
\[ \Big\{ K (X \alpha_1 \ldots \alpha_k) \mathcal{C} \alpha_{k+1} \ldots \alpha_m~|~k \in \{0,\ldots,m-1\} \Big\}. \]

\begin{prop}\label{prop-K}
        Let $x \in L(K (X \alpha_1 \ldots \alpha_k) \mathcal{C} \alpha_{k+1}
        \ldots \alpha_m)$. If $x \to_w y$, then $y \in L(X \alpha_1 \ldots
        \alpha_m)$.
\end{prop}

\begin{proof}
    Let $x = K (X x_1 \ldots x_k) z x_{k+1} \ldots x_m$. Consider its direct reduct $y = X x_1 \ldots x_k x_{k+1} \ldots x_m$. Clearly, $x_i \in L(\alpha_i)$ for $i \in \{1,\ldots,m\}$ which finishes the proof.
\end{proof}

In other words, the set $\KExpansions{\alpha}$ has the property that any \textsc{K-Expansion} of $\alpha$ generates terms that reduce in one step to terms generated by $\alpha$. If we compute the sets $\KExpansions{\alpha}$ for all productions $\alpha \in R_n$, we have almost constructed all of the long $K$-productions in $R_{n+1}$. What remains is to include the production $K R_n \mathcal{C}$ as any term $x \in L(K R_n \mathcal{C})$ reduces directly to $y \in L(\alpha)$ for some production $\alpha \in R_n$.

We use the following subroutine computing the set of \textsc{K-Expansions} of a given production.

\begin{lstlisting}
    -- | Returns K-Expansions of the given production.
    kExpansions :: Tree -> [Tree]
    kExpansions p = case p of
        (K a_1 ... a_m) -> kExpansions' K [a_1,...,a_m]
        (S a_1 ... a_m) -> kExpansions' S [a_1,...,a_m]
        where
            kExpansions' _ [] = []
            kExpansions' h [x_1,...,x_k] = K h C x_1 ... x_k
                : kExpansions' (App h x_1) [x_2,...,x_k]
\end{lstlisting}

\subsection{S-Expansions}\label{sec:s-expansions}
Let us consider a production $\alpha = X \alpha_1 \ldots \alpha_m$ where $m \geq 0$. We would like to define the set $\SExpansions{\alpha}$ similarly to $\KExpansions{\alpha}$, i.e.~in such a way that any term generated by an \textsc{S-Expansion} of $\alpha$ reduces in a single step to some $y \in L(\alpha)$. Unfortunately, defining and computing such a set is significantly more complex than the corresponding $\KExpansions{\alpha}$.

Let $q = X x_1 \ldots x_k z (y z)$. Suppose that $q \in L(\alpha)$ for some production $\alpha \in R_n$. Evidently, $S (X x_1 \ldots x_k) y z \to_w q$ and so we would like to guarantee that $q \in L(\beta)$ for some $\beta \in \SExpansions{\alpha}$. Assume that $\alpha = X \alpha_1 \ldots \alpha_k \gamma \delta$ where $z \in L(\gamma)$ and $yz \in L(\delta)$. Unfortunately, in order to guarantee that we capture all terms reducing to $\alpha$ via an $S$-redex and nothing more, we cannot use both $\gamma$ and $\delta$ directly. We require an additional `rewriting' operation that would extract the important sublanguages of $\gamma$ and $\delta$ so that we can operate on them, instead of $\gamma$ and $\delta$. 

Hence, let us consider the following rewriting relation $\triangleright$, extending the standard derivation relation:

\[ \alpha \triangleright \beta \Leftrightarrow \alpha \to \beta \lor \left(\alpha = \mathcal{C} \land  \exists_{n \in \mathbb{N}}~\beta = R_n \right). \]

We use $\trianglerighteq$ to denote the transitive-reflexive closure of $\triangleright$. The important property of $\trianglerighteq$ is the fact that if $\alpha \trianglerighteq \beta$, then $L(\beta) \subseteq L(\alpha)$. To denote the fact that $\alpha$ does not rewrite to $\beta$ and vice versa, we use the symbol $\alpha \parallel \beta$. In such case we say that $\alpha$ and $\beta$ are \emph{non-rewritable}. Otherwise, if one of them rewrites to the other, meaning that $\alpha$ and $\beta$ are \emph{rewritable}, we use the symbol $\alpha \bowtie \beta$.

\subsubsection{Mesh Set}
In the endeavour of finding appropriate \textsc{S-Expansions} rewritings, we need to find common \emph{meshes} of given non-rewritable trees $\alpha \parallel \beta$. In other words, a complete partition of $L(\alpha) \cap L(\beta)$ using all possible trees $\gamma$ such that $\alpha,\beta \trianglerighteq \gamma$. For this purpose, we use the following pseudo-code subroutines.

\begin{lstlisting}[mathescape=true]
    -- | Given X $\alpha_1 \ldots \alpha_m$ and X $\beta_1 \ldots \beta_m$ computes
    -- the family $\{\gamma_1, \ldots, \gamma_m\}$ of tree meshes. 
    mesh :: [Tree] -> [Tree] -> [[Tree]]
    mesh (x : xs) (y : ys)
        | x `rew` y = [y] : mesh xs ys -- case when x $\trianglerighteq$ y
        | y `rew` x = [x] : mesh xs ys -- case when y $\trianglerighteq$ x
        | otherwise = meshSet x y : mesh xs ys -- case when x $\parallel$ y
    mesh [] [] = []
\end{lstlisting}

The function \textsc{Mesh}, when given two similar productions $\alpha = X \alpha_1 \ldots \alpha_m$ and $\beta = X \beta_1 \ldots \beta_m$, constructs a family $\{\gamma_i\}_{i = 1}^{m}$ where each $\gamma_i$ depends on the comparison of corresponding arguments. In the case when $x$ rewrites to $y$ (denoted as \verb|x `rew` y| in the pseudo-code) the singleton $\{y\}$ is constructed. Similarly, when $y \trianglerighteq x$, the singleton $\{x\}$ is constructed. Otherwise, when $x$ and $y$ are both non-rewritable, $\gamma_i$ is computed using the \textsc{MeshSet} subroutine.

\begin{lstlisting}
    -- | Returns the mesh set of given trees.
    meshSet :: Tree -> Tree -> [Tree]
    meshSet (X a_1 ... a_m) (X b_1 ... b_m) =
    	cartesian X [mesh a_i b_i | i <- [1..m]]
    meshSet (R k) b @ (X b_1 ... b_m) = 
    	nub $ concatMap (\p -> meshSet p b) $ productions (R k)
    meshSet b @ (X b_1 ... b_m) (R k) =
    	nub $ concatMap (\p -> meshSet b p) $ productions (R k)
    meshSet _ _ = []		
\end{lstlisting}

When given two similar trees $\alpha = X \alpha_1 \ldots \alpha_m$ and $\beta = X \beta_1 \ldots \beta_m$, \textsc{MeshSet} computes meshes $\gamma_1,\ldots,\gamma_m$ of corresponding arguments $\alpha_i$ and $\beta_i$ using the subroutine \textsc{Mesh}. Next, argument meshes $\{\gamma_i\}_{i = 1}^{m}$ are used to construct meshes for $\alpha$ and $\beta$, using the subroutine \textsc{Cartesian} which computes the Cartesian product $\{X\} \times \gamma_1 \times \cdots \times \gamma_m$ using term application.
In the case when one of \textsc{MeshSet}'s argument is a reduction grammar $R_k$ and the other $\alpha$ is complex, \textsc{MeshSet} computes recursively mesh sets of $\alpha$ and each production $\delta \in R_k$, outputting their set-theoretic union. In any other case, \textsc{MeshSet} returns the empty set.

\begin{example}
Let $\alpha = K \mathcal{C} R_0 S$ and $\beta = K S (S R_0 \mathcal{C}) S$. Consider $\MeshSet{\alpha}{\beta}$. Both $\alpha$ and $\beta$ are similar and complex, hence
\textsc{MeshSet} proceeds directly to construct mesh sets of corresponding arguments
of $\alpha$ and $\beta$. Since $\mathcal{C} \trianglerighteq S$, we get $\gamma_1 = \{S\}$. Then, as both $R_0$ and $S R_0 \mathcal{C}$ are non-rewritable, $\gamma_2 = \MeshSet{R_0}{S R_0 \mathcal{C}}$. It follows that $\MeshSet{R_0}{S R_0 \mathcal{C}}$
is equal to $\bigcup_{\delta \in R_0} \MeshSet{\delta}{S R_0 \mathcal{C}}$. Further inspection reveals that
$\MeshSet{R_0}{S R_0 \mathcal{C}} = \{S R_0 R_0\}$ and thus $\gamma_2 = \{S R_0 R_0\}$. Finally, $\gamma_3 = \{S\}$ as $S$ rewrites trivially to itself. Since each $\gamma_i$ is a singleton, it follows that
\[ \MeshSet{\alpha}{\beta} = \{K S (S R_0 R_0) S\}. \]
\end{example}

We leave the analysis of \textsc{MeshSet}  until we fully define the construction of reduction grammars ${\{R_n\}}_{n \in \mathbb{N}}$.

\subsubsection{Rewriting Set}
Consider again our previous example of $q = X x_1 \ldots x_k z (y z) \in L(\alpha)$ where $\alpha = X \alpha_1 \ldots \alpha_k \gamma \delta$ such that both $z \in L(\gamma)$ and $yz \in L(\delta)$. In order to capture terms reducing to $\alpha$ via an $S$-redex, we need to find all pairs of trees $\eta,\zeta$ such that $\gamma \trianglerighteq \zeta$ and $\delta \trianglerighteq \eta\, \zeta$. Since such pairs of trees follow exactly the structure of $z (y z)$ we can use them to define the set $\SExpansions{\alpha}$. And so, to find such rewriting pairs, we use the following \textsc{RewritingSet} pseudo-code subroutine.

\begin{lstlisting}[escapeinside={(*}{*)}]
    -- | Given (*$\alpha$*) and (*$\beta$*) computes their rewriting set. 
    rewritingSet :: Tree -> Tree -> [Tree]
    rewritingSet a S = []
    rewritingSet a K = []
    rewritingSet a C = [C a]
    rewritingSet a (R k) = 
    	nub $ concatMap (\p -> rewritingSet a p) $ productions (R k)
    rewritingSet a (X b_1 ... b_m)
    	| a `rew` b_m => [X b_1 ... b_m]
    	| b_m `rew` a => [X b_1 ... b_{m-1} a]
    	| otherwise => 
    		cartesian (X b_1 ... b_{m-1}) [meshSet a b_m] 
\end{lstlisting}

The outcome of \textsc{RewritingSet}($\alpha,\beta$) depends on $\beta$'s structure. If $\beta$ is a primitive combinator $S$ or $K$, \textsc{RewritingSet} returns the empty set. If $\beta = \mathcal{C}$, a singleton $\{\mathcal{C} \alpha\}$ is returned. When $\beta = R_k$ for some $k\in \mathbb{N}$, \textsc{RewritingSet} computes recursively the rewriting sets of $\alpha$ and $\gamma \in R_k$, outputting their set-theoretic union. Otherwise when $\beta = X \beta_1 \ldots \beta_m$, \textsc{RewritingSet} determines whether $\alpha \bowtie \beta_m$. If $\alpha \trianglerighteq \beta_m$, a singleton $\{X \beta_1, \ldots, \beta_m\}$ is returned. Conversely, in the case of $\beta_m \trianglerighteq \alpha$, \textsc{RewritingSet} returns $\{X \beta_1, \ldots, \beta_{m-1} \alpha\}$. Finally if $\alpha$ and $\beta_m$ are non-rewritable, \textsc{RewritingSet} invokes the \textsc{Cartesian} subroutine computing the Cartesian product of $\{X \beta_1, \ldots, \beta_{m-1}\} \times \textsc{MeshSet}(\alpha,\beta_m)$ using term application, passing afterwards its result as the computed rewriting set.

\begin{example}
Let us consider the rewriting set $\RewritingSet{S}{R_0}$. Since $\beta = R_0$, we know that $\RewritingSet{S}{R_0} = \bigcup_{\gamma \in R_0} \RewritingSet{S}{\gamma}$.
It follows therefore that in order to compute $\RewritingSet{S}{R_0}$, we have to consider rewriting sets involving productions of $R_0$. Note that both productions $S$ and $K$ do not contribute new trees. It remains to consider productions $S R_0$, $K R_0$ and $S R_0 R_0$. Evidently, each of them is complex and has $R_0$ as its final argument. Hence, their corresponding rewriting sets are $S S$, $K S$ and $S R_0 S$, respectively. And so, we obtain that
\[ \RewritingSet{S}{R_0} = \{S S, K S, S R_0 S\}. \]
\end{example}

Similarly to the case of \textsc{MeshSet}, we postpone the analysis until we define the construction of ${\{R_n\}}_{n \in \mathbb{N}}$.

Equipped with the notion of mesh and rewriting sets, we are ready to define the set of \textsc{S-Expansions}. And so, let $\alpha = X \alpha_1 \ldots \alpha_m$ where $m \geq 0$. The set $\SExpansions{\alpha}$ is defined as
\[ \Big\{ S(X \alpha_1 \ldots \alpha_k) \varphi_l
                \varphi_r \alpha_{k+3} \ldots \alpha_m~|~k \in \{0,\ldots,m-2\} \Big\}, \]
where $(\varphi_l \varphi_r) \in \RewritingSet{\alpha_{k+1}}{\alpha_{k+2}}$. We use the following subroutine computing the set of \textsc{S-Expansions} for a given $\alpha$.

\begin{lstlisting}
  -- | Returns S-Expansions of the given production.
  sExpansions :: Tree -> [Tree]
  sExpansions p = case p of
      (K a_1 ... a_m) -> sExpansions' K [a_1,...,a_m]
      (S a_1 ... a_m) -> sExpansions' S [a_1,...,a_m]
      where
          sExpansions' _ [] = []
          sExpansions' _ [_] = []
          sExpansions' h [x_1,x_2,...,x_k] =
                map (\(App l r) -> S h l r x_3 ... x_m) 
                (rewritingSet x_1 x_2) ++
                     sExpansions' (App h x_1) [x_2,...,x_m]
\end{lstlisting}

\begin{prop}\label{prop-S}
 Let $x \in L(S(X \alpha_1 \ldots \alpha_k) \varphi_l
 \varphi_r \alpha_{k+3} \ldots \alpha_m)$. If $x \to_w y$, then\\ $y \in L(X
 \alpha_1 \ldots \alpha_k \varphi_r (\varphi_l\, \varphi_r) \alpha_{k+3} \ldots
 \alpha_m)$.
\end{prop}

\begin{proof}
    Let $x = S (X x_1 \ldots x_k) w z x_{k+3} \ldots x_m$. Let us consider its direct reduct $y$ in form of $X x_1 \ldots x_k z (w\, z) x_{k+3} \ldots x_m$. Clearly, $x_i \in  L(\alpha_i)$ for $i$ in proper range. Moreover, both $w \in L(\varphi_l)$ and $z \in L(\varphi_r)$, which finishes the proof.
\end{proof}

\subsection{Algorithm pseudo-code}\label{sec:pceudo-code}
With the complete and formal definitions of both \textsc{S-} and \textsc{K-Expansions} we are ready to give the main algorithm \textsc{Reduction Grammar}, which for given $n \in \mathbb{N}$ constructs the grammar $R_n$.

\begin{lstlisting}[escapeinside={(*}{*)}]
  -- | Given (*$n \in \mathbb{N}$*) constructs (*$R_n$*).
  reductionGrammar :: Integer -> [Tree]
  reductionGrammar 0 = [S, K, S (R 0), K (R 0), S (R 0) (R 0)]
  reductionGrammar n = [S (R n), K (R n)]
  	++ [S (R $ n-i) R_i | i <- [0..n]]
  	++ [K (R $ n-1) C]
  	++ concatMap kExpansions (reductionGrammar $ n-1)
  	++ concatMap sExpansions (reductionGrammar $ n-1)
\end{lstlisting}

\begin{example}
Let us consider $\alpha = S S S R_0$. Since $\alpha \in \SExpansions{S R_0 R_0}$ we get $\alpha \in R_1$. Note that $\SExpansions{\alpha}$ contains $\beta_1 = S (S S) S S$ and
 $\beta_2 = S (S S ) K S$. It follows that $\beta_1, \beta_2 \in R_{2}$.
\end{example}

\section{Analysis}\label{sec:analysis}
\subsection{Tree potential}\label{sec:tree-potential}
Most of our proofs in the following sections are using inductive reasoning on the underlying tree structure. Unfortunately, in certain cases most natural candidates for induction such as tree size fail due to \emph{self-referencing productions}, i.e.~productions of $R_n$ which explicitly use the non-terminal symbol $R_n$. In order to remedy such problems, we introduce the notion of \emph{tree potential} $\potential{\alpha}$, defined inductively as\newpage
\[ \potential{S} = \potential{K} = \potential{\mathcal{C}} = 0,\]
\[ \potential{X \alpha_1 \ldots \alpha_m} = m + \sum_{i=1}^{m}
\potential{\alpha_i}, \]
\[ \potential{R_n} = 1 + \max_{\gamma \in \Phi (R_n)} \potential{\gamma} \]
where $\Phi (R_n)$ denotes the set of productions of $R_n$ which do not use the non-terminal symbol $R_n$. Note that such a definition of potential is almost identical to the notion of tree size. The potential of $\alpha$ is the sum of $\alpha$'s size and the weighted sum of all non-terminal grammar symbols occurring in $\alpha$. 

Immediately from the definition we get $\potential{R_0} =
1$. Moreover, $\potential{R_{n+1}} > \potential{R_n}$ for any $n \in \mathbb{N}$.  Indeed, let $\alpha \in R_n$ be the witness of $R_n$'s potential. Clearly, $(K \alpha\, \mathcal{C}) \in \Phi(R_{n+1})$ and so $R_{n+1}$ has necessarily greater potential. Moreover, $\potential{\alpha} > \potential{\beta}$ if $\beta$ is a subtree of $\alpha$. It follows that the notion of tree potential is a good candidate for the intuitive tree complexity measure.

\subsection{Soundness}\label{sec:soundness}

In this section we are interested in the soundness of \textsc{Reduction Grammar}. In particular, we prove that it is computable, terminates on all legal inputs and, for given $n$, constructs a reduction grammar $R_n$ generating only terms that require exactly $n$ steps to normalize.

Let us start with showing that the rewriting relation is decidable.

\begin{prop}
    It is decidable to check whether $\alpha \trianglerighteq \beta$.
\end{prop}

\begin{proof}
        Induction over $n = \potential{\alpha} + \potential{\beta}$.  If $\alpha
        = X$, then the only tree $\alpha$ rewrites to is $X$. On the other hand,
        if $\alpha = \mathcal{C}$, then $\alpha$ rewrites to any $\beta$. And so,
        it is decidable to check whether $\alpha \trianglerighteq \beta$ in case
        $n = 0$. Now, let us assume that $n > 0$. We have two remaining cases to
        consider.
        \begin{enumerate}[(i)]
            \item If $\alpha = X \alpha_1 \ldots \alpha_m$, then $\alpha
                    \trianglerighteq \beta$ if and only if $\beta = X \beta_1
                    \ldots \beta_m$ and $\alpha_i \trianglerighteq \beta_i$ for
                    all $i \in \{1,\ldots,m\}$. Since the total potential of
                    $\potential{\alpha_i} + \potential{\beta_i}$ is less than $n$, we can use
                    the induction hypothesis to decide whether all arguments of
                    $\alpha$ rewrite to the respective arguments of $\beta$. It
                    follows that we can decide whether $\alpha \trianglerighteq
                    \beta$.
           \item If $\alpha = R_k$, then clearly $\alpha \trianglerighteq \beta$
            if and only if $\beta = R_k$ or there exists a production $\gamma \in 
            	R_k$ such that
                   $\gamma \trianglerighteq \beta$. Let us assume that $\gamma$
                   is a production of $R_k$. Note that if $\gamma
                   \trianglerighteq \beta$, then $\gamma$ and $\beta$ are
                   similar. And so, since similarity is decidable, we can
                   rephrase our previous observation as $\alpha \trianglerighteq
                   \beta$ if and only if $\beta = R_k$ or there exists a production $\gamma \in
                   R_k$ such that $\gamma$ is similar to $\beta$ and $\gamma
                   \trianglerighteq \beta$. Checking whether $\beta = R_k$ is trivial, so let us assume the other option and start with the case when
                   $\gamma$ is a short production referencing $R_k$.

                   If $\gamma = X R_k$ is similar to $\beta = X \beta_1$, we
                   know that $\gamma \trianglerighteq \beta$ if and only if $R_k
                   \trianglerighteq \beta_1$. Since $\potential{R_k} +
                   \potential{\beta_1} < n$, we know that checking whether $R_k
                   \trianglerighteq \beta_1$ is decidable, hence so is 
                   $\gamma \trianglerighteq \beta$.
                   
                   Let us assume w.l.o.g.~that $\gamma = S R_k R_0$. Clearly,
                   $\beta = S \beta_1 \beta_2$. And so, $\gamma \trianglerighteq
                   \beta$ if and only if $R_k \trianglerighteq \beta_1$ and $R_0
                   \trianglerighteq \beta_2$. Notice that $\potential{R_k} +
                   \potential{\beta_1} < n$ as well as $\potential{R_0} +
                   \potential{\beta_2} < n$. Using the induction hypothesis to
                   both, we get that checking $R_k \trianglerighteq \beta_1$ and
                   $R_0 \trianglerighteq \beta_2$ is decidable, hence so is $\alpha
                   \trianglerighteq \beta$.

                   Finally, if $\gamma$ is a long production we can rewrite it
                   as $\gamma = X \gamma_1 \ldots \gamma_m$, and so reduce this case to
                   the previous one when both trees are complex, as $\potential{\gamma}$ is necessarily smaller than $n$.
        \end{enumerate}
\end{proof}

\begin{prop}\label{prop-meshset-rewrites}
    Let $\alpha,\beta$ be two trees.  Then, both $\alpha \trianglerighteq
    \gamma$ and $\beta \trianglerighteq \gamma$ for arbitrary $\gamma \in
    \MeshSet{\alpha}{\beta}$.
\end{prop}

\begin{proof}
        Induction over $n = \potential{\alpha} + \potential{\beta}$. Let $M =
        \MeshSet{\alpha}{\beta}$.  Clearly, it suffices to consider such
        $\alpha,\beta$ that $M \neq \emptyset$.
        
        Let us assume that both $\alpha = X \alpha_1 \ldots \alpha_m$ and $\beta
        = X \beta_1 \ldots \beta_m$. If $\alpha_i \bowtie \beta_i$ for all $i \in
        \{1,\ldots,m\}$, then $M$ consists of a single tree $\gamma = X \gamma_1
        \ldots \gamma_m$ for which $\alpha_i,\beta_i \trianglerighteq \gamma_i$.
        Evidently, our claim holds. Suppose that there exists an $i \in
        \{1,\ldots,m\}$ such that $\alpha_i \parallel \beta_i$. Since
        $\potential{\alpha_i} + \potential{\beta_i} < n$, we can apply the
        induction hypothesis to $\MeshSet{\alpha_i}{\beta_i}$.  The set
        $M' = \MeshSet{\alpha_i}{\beta_i}$ cannot be empty and so let $\delta_i$ be
        an arbitrary mesh in $M'$. We know that $\alpha_i,\beta_i
        \trianglerighteq \delta_i$. And so, if we consider an arbitrary $\gamma
        = X \gamma_i \ldots \gamma_m \in M$, we get $\alpha_i,\beta_i
        \trianglerighteq \gamma_i$ for all $i \in \{1,\ldots,m\}$, which implies
        our claim.

        What remains is to consider the case when either $\alpha = R_k$ and
        $\beta$ is complex or, symmetrically, $\beta = R_k$ and $\alpha$ is
        complex. Let us
        assume w.l.o.g.~the former case.
        From the definition, $\MeshSet{R_k}{\beta}$ depends
        on the union of $\MeshSet{\gamma}{\beta}$ for $\gamma \in R_k$. Clearly,
        $R_k$ rewrites to any of its productions. Let $\gamma \in R_k$ be a
        production referencing $R_k$. We have to consider two cases based on
        the structure of $\gamma$.
        \begin{enumerate}[(i)]
                \item   Let $\gamma = X R_k$. Then, $\potential{\gamma} =
                        \potential{R_k} + 1$ and so we cannot use the induction
                        hypothesis to $\MeshSet{\gamma}{\beta}$ directly. Note
                        however, that we can assume that $\beta = X \beta_1$,
                        since otherwise $\MeshSet{\gamma}{\beta}$ would be
                        empty. Therefore, we know that $\MeshSet{R_k}{\beta_1} \neq
                        \emptyset$ to which we can now use the induction hypothesis,
                        as $\potential{R_k} + \potential{\beta_1} < n$.
                        Immediately, we get that $R_k,\beta \trianglerighteq
                        \gamma$. 
                \item W.l.o.g.~let $\gamma = S R_k R_{0}$.  Then,
                        $\potential{\gamma} = 3 + \potential{R_k}$. Again, we
                        cannot directly use the induction hypothesis. Note
                        however, that we can assume that $\beta = S \beta_1
                        \beta_2$. And so we get $\potential{R_k} +
                        \potential{\beta_1} < n$ and $\potential{R_0} +
                        \potential{\beta_2} < n$. Using the induction hypothesis
                        to both parts we conclude that $R_k,\beta
                        \trianglerighteq \gamma$ in this case as well.
        \end{enumerate}

        To finish the proof we need to show that our claim holds for all $\gamma \in
        R_k$ which do not reference $R_k$. Indeed, any such production has
        necessarily smaller potential than $R_k$, and so, we can use the
        induction hypothesis directly to the resulting mesh set. Evidently, our
        claim holds.
\end{proof}

In other words, $\MeshSet{\alpha}{\beta}$ is in fact a set of \emph{meshes}, i.e.~trees generating a joint portion of $L(\alpha)$ and $L(\beta)$. Note, that along the lines of proving the above proposition, we have also showed that indeed $\MeshSet{\alpha}{\beta}$ terminates on all legal inputs, as the number of recursive calls cannot exceed $2 (\potential{\alpha} + \potential{\beta})$ -- in the worst case, every second recursive call decreases the total potential sum of its inputs.

\begin{prop}\label{prop-rewritingset-rewrites}
    Let $\alpha,\beta$ be two trees. Then, $\alpha \trianglerighteq
    \varphi_r$ and $\beta \trianglerighteq \varphi_l \varphi_r$ for arbitrary
    $\varphi_l \varphi_r \in \RewritingSet{\alpha}{\beta}$.
\end{prop}

\begin{proof}
    We can assume that $\RewritingSet{\alpha}{\beta} \neq \emptyset$, as otherwise our claim trivially holds.
    Let $\varphi_l \varphi_r \in \RewritingSet{\alpha}{\beta}$. Based on the structure of $\beta$, we have to three cases to consider.
    \begin{enumerate}[(i)]
        \item If $\beta = \mathcal{C}$, then $\varphi_l \varphi_r = \mathcal{C}
                \alpha$. Clearly, $\alpha \trianglerighteq
                    \alpha$ and $\mathcal{C} \trianglerighteq \mathcal{C}\,
                    \alpha$.
        \item If $\beta = X \beta_1 \ldots \beta_m$, then we have again exactly three possibilities. Both cases when $\alpha \bowtie \beta_m$ are
                trivial, so let us assume that $\alpha \parallel \beta_m$. It
                follows that there exists such a $\gamma \in
                \MeshSet{\alpha}{\beta_m}$ that $\varphi_l \varphi_r = X \beta_1
                \ldots \beta_{m-1} \gamma$. Due to
                Proposition~\ref{prop-meshset-rewrites}, we know that
                $\alpha,\beta_m \trianglerighteq \gamma$ and so directly that
                $\alpha \trianglerighteq \varphi_r$ and 
                $\beta \trianglerighteq \varphi_l \varphi_r$.
        \item Finally, suppose that $\beta = R_n$. Then, there exists a production $\gamma \in
                R_n$ such that $\varphi_l \varphi_r \in \RewritingSet{\alpha}{\gamma}$.
                Note however, that in this case $\gamma = X \gamma_1 \ldots
                \gamma_m$ and so we can reduce this case to the already considered case above.
    \end{enumerate}
\end{proof}

Now we are ready to give the anticipated soundness theorem.

\begin{theorem}[Soundness]\label{the-correctness}
        If $x \in L(R_n)$, then $x$ reduces in $n$ steps.
\end{theorem}
\begin{proof}
        Induction over pairs $(n,m)$ where $m$ denotes the length of a
        minimal, in terms of length, derivation $\Sigma$ of $x \in L(R_n)$.  Let
        $n = 0$ and so $x \in L(R_0)$. If $m = 1$, then $x \in \{S,K\}$ hence
        $x$ is already in normal form. Suppose that $m > 1$. Clearly, $x \not \in
        \{S,K\}$. Let $R_0 \to \alpha$ be the first production rule used in
        derivation $\Sigma$. Using the induction hypothesis to the reminder of
        the derivation, we know that $x$ does not contain any nested redexes.
        Moreover, $\alpha$ avoids any head redexes and so we get that $x$ is in
        normal form.  
        
        Let $n > 0$. We have to consider several cases based on the
        choice of the first production rule $R_n \to \alpha$ used in the
        derivation $\Sigma$.
        \begin{enumerate}[(i)]
            \item $\alpha = S R_n$ or $\alpha = K R_n$. Using the
                    induction hypothesis we know that $x = X y$ where $y$
                    reduces in $n$ steps. Clearly, so does $x$.
            \item $\alpha = S R_{n-i} R_i$ for some $i \in \{0,\ldots,n\}$.
                    Then, $x = S y z$ where $y \in L(R_{n-i})$ and $z \in L(R_i)$.
                    Note that both their derivations are in fact shorter than the
                    derivation of $x$ and thus applying the induction hypothesis
                    to both $y$ and $z$ we know that they reduce in $n-i$ and $i$ steps,
                    respectively. Following the normal-order reduction strategy,
                    we note that $y$ and $z$ and reduce sequentially in $x$.
                    Since $x$ does not contain a head redex itself, we reduce it
                    in total of $n$ reductions.
            \item $\alpha = K R_{n-1} \mathcal{C}$. Directly from the induction
                    hypothesis we know that $x = K y z$ where $y$ reduces in
                    $n-1$ steps. And so $x \to_w y$, implying that $x$ reduces
                    in $n$ steps.
            \item $\alpha = K(X \alpha_1 \ldots \alpha_k) \mathcal{C}
                    \alpha_{k+1} \ldots \alpha_m$. Let $x \in L(\alpha)$.
                    Clearly, $x$ has a head redex and so let $x \to_w y$. Using
                    Proposition~\ref{prop-K}, we know that $y \in L(X \alpha_1
                    \ldots \alpha_m)$. Moreover, by the construction of $R_n$
                    we get $\alpha \in \KExpansions{X \alpha_1 \ldots \alpha_m}$
                    and therefore $y \in L(R_{n-1})$. It follows that $y$
                    reduces in $n-1$ steps and so $x$ in $n$ steps.
            \item $\alpha = S(X \alpha_1 \ldots \alpha_k) \varphi_l \varphi_r
                    \alpha_{k+3} \ldots \alpha_m$. Let $x \in L(\alpha)$.
                    Clearly, $x$ has a head redex and so let $x \to_w y$.
                    Due to Proposition~\ref{prop-S} we get that $y \in L(X
                    \alpha_1 \ldots \alpha_k \varphi_r (\varphi_l\, \varphi_r) \alpha_{k+3} \ldots
                    \alpha_m)$. In order to show that $x$ reduces in $n$ steps
                    it suffices to show that $y \in L(R_{n-1})$. Let us consider
                     $\beta$ such that $\alpha \in \SExpansions{\beta}$. From the
                     structure of $\alpha$ we can rewrite
                     it as $\beta = X \alpha_1 \ldots \alpha_k \alpha_{k+1}
                     \alpha_{k+2} \ldots \alpha_{m}$. Moreover, from
                    Proposition~\ref{prop-rewritingset-rewrites} we know that
                    $\alpha_{k+1} \trianglerighteq \varphi_r$ and $\alpha_{k+2}
                    \trianglerighteq \varphi_l\, \varphi_r$. Clearly, $y \in
                    L(\beta)$, which finishes the proof.
        \end{enumerate}
\end{proof}

Combining the above result with the fact that each normalizing combinatory logic term reduces in a determined number of normal-order reduction steps, gives us the following corollary.
\begin{cor}\label{cor-Rn-diff-Rm}
        If $L(R_n) \cap L(R_m) \neq \emptyset$, then $n = m$.
\end{cor}

\subsection{Completeness}\label{sec:completeness}
In this section we are interested in the completeness of \textsc{Reduction Grammar}. In other words, we show that every term normalizing in exactly $n$ steps is generated by $R_n$.

W start with some auxiliary lemmas showing the completeness of \textsc{MeshSet} and, in consequence, \textsc{RewritingSet}.

\begin{lemma}\label{lem-meshset-crutial}
        Let $\alpha,\beta$ be two non-rewritable trees.
        Let $x$ be a term. Then, $x \in
        L(\alpha) \cap L(\beta)$ if and only if there exists a mesh $\gamma
        \in \MeshSet{\alpha}{\beta}$ such that $x \in L(\gamma)$.
\end{lemma}

\begin{proof}
    It suffices to show the necessary part, the sufficiency is clear from
    Proposition~\ref{prop-meshset-rewrites}. We show this result using
    induction over the size $|x|$ of $x$. Let $x \in L(\alpha) \cap L(\beta)$.
    Let us start with noticing that $|\alpha| + |\beta| > 0$. Moreover, there
    are only two cases where $x \in L(\alpha) \cap L(\beta)$, i.e.~when either $\alpha = X \alpha_1 \ldots \alpha_m$ and $\beta = X \beta_1
    \ldots \beta_m$ or when exactly one of them is equal to some $R_n$ and the other
    is complex. And so, let us consider these cases separately.
    \begin{enumerate}[(i)]
        \item Suppose that $\alpha = X \alpha_1 \ldots \alpha_m$ and $\beta = X \beta_1
                \ldots \beta_m$. It follows that we can rewrite $x$ as $X x_1
                \ldots x_m$ such that $x_i \in L(\alpha_i) \cap L(\beta_i)$.
                Clearly, if all $\alpha_i \bowtie \beta_i$, then there exists a mesh
                $\gamma$ such that $x \in L(\gamma)$. Let us assume that some
                $\alpha_i$ and $\beta_i$ are non-rewritable. Then, using
                the induction hypothesis we find a mesh $\gamma_i \in
                \MeshSet{\alpha_i}{\beta_i}$ such that $x_i \in L(\gamma_i)$.
                Immediately, we get that there exists a mesh in $\MeshSet{\alpha}{\beta}$
                which generates $x$.
        \item Let us assume w.l.o.g.~that $\alpha = R_n$ and $\beta = X \beta_1
                \ldots \beta_m$. Since $x \in L(R_n)$, there must be such a
                production $\gamma \in R_n$ that $x \in L(\gamma)$. Although the
                size of $x$ does not decrease, note that we can reduce this case
                to the one considered above since both $\gamma$ and $\beta$ are
                complex. Clearly, it follows that we can find a suiting
                mesh $\delta \in \MeshSet{\gamma}{\beta}$ such that $x \in
                L(\delta)$. Immediately, we get $\delta \in
                \MeshSet{\alpha}{\beta}$ which finishes the proof.
    \end{enumerate}
\end{proof}

\begin{lemma}\label{lem-rewritingset-crutial}
    Let $\alpha,\beta$ be two trees. Let $x, yx$ be two terms. Then,
    $x \in L(\alpha)$ and $y x \in L(\beta)$ if and only if there exists
    such a $\varphi_l \varphi_r \in \RewritingSet{\alpha}{\beta}$
    that $x \in L(\varphi_r)$ and $y x \in L(\varphi_l \varphi_r)$.
\end{lemma}

\begin{proof}
        Due to Proposition~\ref{prop-rewritingset-rewrites} the sufficiency part
        is clear. What remains is to show the necessary part. Let $x \in
        L(\alpha)$ and $y x \in L(\beta)$. Consider the structure of
        $\beta$. If $\beta = \mathcal{C}$, then $\mathcal{C} \alpha \in
        \RewritingSet{\alpha}{\beta}$ and so $\varphi_l = \mathcal{C}, \varphi_r
        = \mathcal{\alpha}$.  Clearly, our claim holds. Now, consider the case
        when $\beta = X \beta_1 \ldots \beta_m$. Based on the rewritability of
        $\alpha$ and $\beta_m$ we distinguish three subcases.
        \begin{enumerate}[(i)]
            \item If $\alpha \trianglerighteq \beta_m$, then $X \beta_1 \ldots
                    \beta_m \in \RewritingSet{\alpha}{\beta}$. Since $y x 
                    \in L(\beta)$, we get $x \in L(\beta_m)$ and in
                     consequence $x \in L(\varphi_r)$.
            \item If $\beta_m \trianglerighteq \alpha$, then $X \beta_1 \ldots
                    \beta_{m-1} \alpha \in \RewritingSet{\alpha}{\beta}$. Since
                    $\beta_m \trianglerighteq \alpha$, we know that $L(\alpha)
                    \subseteq L(\beta_m)$ and so $y x \in 
                    L(X \beta_1 \ldots \beta_{m-1} \alpha)$.
            \item If $\alpha \parallel \beta_m$, then we know that $x \in
                    L(\alpha) \cap L(\beta_m)$. If not, then $y x$ could not be
                    a term of $L(\beta)$. And so, using Lemma~\ref{lem-meshset-crutial}
                    we find a
                    mesh $\gamma \in \MeshSet{\alpha}{\beta_m}$ such that $x \in
                    L(\gamma)$. We know that $X \beta_1 \ldots \beta_{m-1}
                    \gamma \in \RewritingSet{\alpha}{\beta}$. Clearly, it is the
                    tree we were looking for.
        \end{enumerate}

        It remains to consider the case when $\beta = R_k$. Note however, that
        it can be reduced to the case when $\beta = X \beta_1 \ldots \beta_m$.
        Indeed, since $x \in L(R_k)$, then there exists a production $\gamma \in
        R_k$ such that $x \in L(\gamma)$. From the previous arguments we
        know that we can find a tree satisfying our claim.
\end{proof}

Using the above completeness results for \textsc{MeshSet} and \textsc{RewritingSet}, we are ready to give the anticipated completeness result of ${\{R_n\}}_{n \in \mathbb{N}}$.

\begin{theorem}[Completeness]\label{the-completness}
    If $x$ reduces in $n$ steps, then $x \in L(R_n)$.
\end{theorem}
\Needspace{3\baselineskip}

\begin{proof}
    Induction over pairs $(n,s)$ where $s$ denotes the size of $x$.
    The base case $n = 0$ is clear due to the completeness of
    $R_0$. Let $n > 0$. 
    
    Let us start with considering short terms.
    Let $x = X y$ be a term of size $s$. Since $x$ has no head redex, $y$ must
    reduce in $n$ steps as well. Now, we can apply the induction hypothesis to
    $y$ and deduce that $y \in L(R_n)$. It follows that $x \in L(X R_n)$.
    Clearly, $X R_n$ is a production of $R_n$ and so $x \in L(R_n)$. Now, assume
    that $x = S y z$. Since $x$ reduces in $n$ steps and does not contain a head
    redex, there exists such an $i \in \{0,\ldots,n\}$ that $y$ reduces in $i$
    steps and $z$ reduces in $n-i$ steps. Applying the induction hypothesis to
    both $y$ and $z$, we get that $y \in L(R_i)$ whereas $z \in L(R_{n-i})$.
    Immediately, we get that $x \in L(R_n)$ as $S R_i R_{n-i} \in R_n$.
   
    What remains is to consider long terms. Let $x = K x_1 x_2$.
    Note that $x_1$ must reduce in $n-1$ steps, as $x \to_w x_1$. And so, from
    the induction hypothesis we get that $x_1 \in L(R_{n-1})$. Now we have
    $x \in L(K R_{n-1} \mathcal{C})$ and hence $x \in L(R_n)$ as $K R_{n-1}
    \mathcal{C}$ is a production of $R_n$.
    
    Now, let $x = K x_1 \ldots x_m$ for $m
    \geq 3$. Since $x$ has a head redex, we know that $x \to_w y = x_1 x_3
    \ldots x_m$, which itself reduces in $n-1$ steps. Let us rewrite $y$ as
	$X y_1 \ldots y_k x_3 \ldots x_m$ where $x_1 = X y_1 \ldots y_k$. We know that
    there exists a production $\alpha \in R_{n-1}$ such that $y \in L(\alpha)$.
    Let $\alpha = X \overline{\alpha_1} \ldots \overline{\alpha_k} \alpha_3 \ldots
    \alpha_m$. Clearly, there exists a $\beta = K (X \overline{\alpha_1} \ldots
    \overline{\alpha_k}) \mathcal{C} \alpha_3 \ldots \alpha_m \in
    \KExpansions{\alpha}$. We claim that $x \in L(\beta)$. Indeed, $y \in
    L(\alpha)$ implies that $y_i \in L(\overline{\alpha_i})$ and $x_j \in
    L(\alpha_j)$ for any $i$ and $j$ in proper ranges. Since $x_2 \in L(\mathcal{C})$, we conclude that $x \in L(\beta)$ and hence $x \in L(R_n)$.

    Let $x = S x_1 \ldots x_m$ for $m \geq 3$. Since $x$ has a head redex $x \to_w y = x_1 x_3 (x_2 x_3) x_4 \ldots x_m$ which reduces in $n-1$
    steps. Again, let us rewrite $y$ as $X y_1 \ldots
    y_k x_3 (x_2 x_3) x_4 \ldots x_m$ where $x_1 = X y_1 \ldots y_k$. Now, since $y \in L(R_{n-1})$, there exists a production $\alpha = X \overline{\alpha_1} 
    \ldots \overline{\alpha_k} \alpha_3 \gamma
    \alpha_4 \ldots \alpha_m \in R_{n-1}$ such that $y \in L(\alpha)$. We claim
    that there must be a production $\beta \in \SExpansions{\alpha}$ such that
    $x \in L(R_n)$. If so, the proof would be complete. Notice that $x_3 \in
    L(\alpha_3)$ and $x_2 x_3 \in L(\gamma)$. Using Lemma~\ref{lem-rewritingset-crutial}
    we know that there exists a tree $\varphi_l \varphi_r \in
    \RewritingSet{\alpha_3}{\gamma}$ such that $x_3 \in L(\varphi_r)$ and $(x_2
    x_3) \in L(\varphi_l \varphi_r)$. And so $y \in L(X \overline{\alpha_1} 
    \ldots \overline{\alpha_k} \varphi_r (\varphi_l \varphi_r)
    \alpha_4 \ldots \alpha_m)$. Moreover, due to the fact that $\varphi_l \varphi_r \in
    \RewritingSet{\alpha_3}{\gamma}$, we know that the tree $\beta = S (X \overline{\alpha_1} 
    \ldots \overline{\alpha_k}) \varphi_l \varphi_r
    \alpha_4 \ldots \alpha_m \in \SExpansions{\alpha}$ and so also $\beta \in
    R_n$. Since $x_2 \in L(\varphi_l)$, we get that $x \in L(\beta)$.
\end{proof}

\subsection{Unambiguity}\label{sec:unambiguity}

In this section we show that reduction grammars are in fact \emph{unambiguous}, i.e.~every term $x \in L(R_n)$ has exactly one derivation. Due to the mutual recursive nature of \textsc{MeshSet}, \textsc{RewritingSet} and \textsc{ReductionGrammar}, we split the proof into two separate parts. In the following lemma, we show that \textsc{MeshSet} returns unambiguous meshes under the assumption that $R_0,\ldots,R_n$ up to some $n$ are  themselves unambiguous. In the corresponding theorem we use inductive reasoning which supplies the aforementioned assumption and thus, as a consequence, allows us to prove the main result.

\begin{lemma}\label{lem-meshset-disjoint}
        Let $\alpha, \beta$ be two trees such that $\gamma,\overline{\gamma} \in
        \MeshSet{\alpha}{\beta}$ where in addition $\rank{\alpha},\rank{\beta}
        \leq r + 1$.  If $R_0,\ldots,R_r$ are unambiguous and $L(\gamma) \cap
        L(\overline{\gamma}) \neq \emptyset$, then $\gamma = \overline{\gamma}$.
\end{lemma}
\Needspace{3\baselineskip}

\begin{proof}
        Induction over $n = \potential{\alpha} + \potential{\beta}$. Let $x \in
        L(\gamma) \cap L(\overline{\gamma})$.  We can assume that
        $|\MeshSet{\alpha}{\beta}|$ is greater than $1$ as the case for
        $|\MeshSet{\alpha}{\beta}| = 1$ is trivial. In consequence, the base case $n = 0$ is clear as the resulting \textsc{MeshSet} for two trees of potential $0$ has to be necessarily empty.  Hence, we have to
        consider two cases based on the structure of $\alpha$ and $\beta$.
        
        \begin{enumerate}[(i)]
        \item Let $\alpha = X \alpha_1 \ldots \alpha_m$ and $\beta = X \beta_1
                \ldots \beta_m$. Clearly, $x$ is in form of $x = X x_1 \ldots
                x_m$.  Let $\alpha_i \parallel \beta_i$ be an arbitrary
                non-rewritable pair of arguments in $\alpha,\beta$.  It follows that $x_i
                \in L(\alpha_i) \cap L(\beta_i)$ and so, due to
                Lemma~\ref{lem-meshset-crutial}, there exists a mesh
                $\delta \in \MeshSet{\alpha_i}{\beta_i}$ such that $x_i \in
                L(\delta)$. Let $M_i = \MeshSet{\alpha_i}{\beta_i}$. Since
                $\potential{\alpha_i} + \potential{\beta_i} < n$ we can use the
                induction hypothesis to $M_i$ and immediately conclude that
                $\delta$ is the only mesh in $M_i$ generating $x_i$. And so, we
                know that $\gamma$ and $\overline{\gamma}$ are equal on the
                non-rewritable arguments of $\alpha,\beta$. Note that if $\alpha_i \bowtie 	\beta_i$, then both contribute a single mesh at position $i$. Immediately, we get that both $\gamma$ and $\overline{\gamma}$ are also equal on the rewritable arguments of $\alpha$ and $\beta$, hence finally $\gamma = \overline{\gamma}$.

        \item W.l.o.g.~let $\alpha = R_k$ and $\beta = X \beta_1 \ldots
                \beta_m$. Clearly, as $\rank{\alpha} \leq r + 1$, we know that $R_k$
                is unambiguous.  From the definition of \textsc{MeshSet} there
                exist productions $\delta,\overline{\delta} \in R_k$ such that
                $\gamma \in \MeshSet{\delta}{\beta}$ and $\overline{\gamma} \in
                \MeshSet{\overline{\delta}}{\beta}$. We claim that $\gamma =
                \overline{\gamma}$ as otherwise $\delta,\overline{\delta}$ would
                generate a common term. Suppose that $\gamma \neq
                \overline{\gamma}$.         
                  From Lemma~\ref{lem-meshset-crutial} we
                know that $L(\gamma) \subseteq L(\delta)$ and $L(\overline{\gamma})
                \subseteq L(\overline{\delta})$. Since $x \in
        		L(\gamma) \cap L(\overline{\gamma})$, we get that $x \in L(\delta)
                \cap L(\overline{\delta})$ and therefore a contradiction with the fact 					that $R_k$ is unambiguous. It follows that $\gamma =
                \overline{\gamma}$, which finishes the proof.
        \end{enumerate}
\end{proof}

\begin{theorem}[Unambiguity]\label{th-unambiguity}
    Let $\alpha, \beta \in R_n$. If $L(\alpha) \cap L(\beta) \neq \emptyset$,
    then $\alpha = \beta$.
\end{theorem}

\begin{proof}
    Induction over $n$. Let $x \in L(\alpha) \cap L(\beta)$. Note that if $x \in L(\alpha) \cap L(\beta)$, then both $\alpha, \beta$ must be similar. We can therefore focus on similar productions of $R_n$. For that reason, we
    immediately notice that $R_0$ satisfies our claim.
    
     Let $n > 0$. Since $R_n$
    does not contain combinators as productions, we can rewrite both
    $\alpha$ as $X \alpha_1 \ldots \alpha_m$ and $\beta$ as $X
    \beta_1 \ldots \beta_m$. Let us consider several cases based on their common
    structure.
    \begin{enumerate}[(i)]
        \item Let $X = K$. If $m = 1$, then $\alpha$ and $\beta$ are
            equal as there is exactly one short $K$-production in $R_n$. If
            $m=2$, then again $\alpha = \beta$, since there is a unique
            $K$-production $K R_{n-1} \mathcal{C}$ of length two in $R_n$. If $m >
            2$, then both are \textsc{K-Expansions} of some productions in
            $R_{n-1}$. And so
            \[ \alpha = K (X \overline{\alpha_1} \ldots \overline{\alpha_k}) \mathcal{C}
            \alpha_3 \ldots \alpha_m \in \KExpansions{\gamma}, \]
            \[ \beta = K (X \overline{\beta_1} \ldots \overline{\beta_k}) \mathcal{C}
            \beta_3 \ldots \beta_m \in \KExpansions{\delta}, \]
            where
            \[\gamma = X \overline{\alpha_1}
            \ldots \overline{\alpha_k} \alpha_3 \ldots \alpha_m,\]
            \[\delta = X \overline{\beta_1}
            \ldots \overline{\beta_k} \beta_3 \ldots \beta_m.\]
            Since $x \in L(\alpha) \cap L(\beta)$, we can assume that $x$ is in
            form of $K (X y_1 \ldots y_k) x_2 x_3 \ldots x_m$ where $y_i \in
            L(\overline{\alpha_i}) \cap L(\overline{\beta_i})$ and $x_j \in L(\alpha_j)
            \cap L(\beta_j)$. It follows that we can use the induction
            hypothesis to $\gamma, \delta \in R_{n-1}$ obtaining
            $\overline{\alpha_i} = \overline{\beta_i}$ and $\alpha_j = \beta_j$.
            Immediately, we get $\alpha = \beta$.
        
        \item Let $X = S$. If $m = 1$, then $\alpha$ and $\beta$ are equal due to
                the fact that there is exactly one $S$-production of length one
                in $R_n$. If $m = 2$, then $\alpha,\beta$ are in form of $\alpha
                = S R_i R_{n-i}$ and $\beta = S R_j R_{n-j}$. Hence, $x = S x_1
                x_2$ for some terms $x_1,x_2$. Since $x_1 \in L(R_i) \cap
                L(R_j)$ and $x_2 \in L(R_{n-i}) \cap L(R_{n-j})$, we know that
                $i = j$ due to Corollary~\ref{cor-Rn-diff-Rm} and thus $\alpha =
                \beta$. It remains to consider long $S$-productions. Let
                \[ \alpha = S (X \overline{\alpha_1} \ldots \overline{\alpha_k}) \varphi_l
                        \varphi_r \alpha_4
                    \ldots \alpha_m \in \SExpansions{\gamma},\]
                    \[ \beta = S (X \overline{\beta_1} \ldots \overline{\beta_k})
                            \overline{\varphi_l} \overline{\varphi_r} \beta_4
                    \ldots \beta_m \in \SExpansions{\delta},\]
                where
                \[ \gamma = X \overline{\alpha_1} \ldots \overline{\alpha_k} \alpha_2
                \alpha_3 \alpha_4 \ldots \alpha_m, \]
                \[ \delta = X \overline{\beta_1} \ldots \overline{\beta_k} \beta_2
                \beta_3 \beta_4 \ldots \beta_m. \]
                It follows that we can rewrite $x$ as $S (X y_1 \ldots y_k) w z
                x_4 \ldots x_m$. Let us focus on the reduct $x \to_w y = X y_1
                \ldots y_k z (w z) x_4 \ldots x_m$. Evidently, $y \in L(\gamma)
                \cap L(\delta)$ and so according to the induction hypothesis we
                know that $\gamma = \delta$, in particular $\alpha_2 = \beta_2$
                and $\alpha_3 = \beta_3$. Hence, both $\varphi_l \varphi_r$ and
                $\overline{\varphi_l}\overline{\varphi_r}$ are elements of the same
                \textsc{RewritingSet}. If we could guarantee that $\varphi_l
                \varphi_r = \overline{\varphi_l}\overline{\varphi_r}$, then immediately
                $\alpha = \beta$ and the proof is finished.  From the
                construction of the \textsc{RewritingSet} we have two cases left to
                consider.
                    \begin{enumerate}[(i)]
                        \item If $\alpha_3 = X \gamma_1 \ldots \gamma_m$, then
                                both $\varphi_l \varphi_r$ and $\overline{\varphi_l}
                                \overline{\varphi_r}$ are either in form of $X \gamma_1
                                \ldots \gamma_{m-1}\,
                                \varphi_r$ or $X \gamma_1
                                \ldots \gamma_{m-1}\,
                                \overline{\varphi_r}$.
                                It follows that
                                $\varphi_l = \overline{\varphi_l}$. It remains to
                                show that $\varphi_r = \overline{\varphi_r}$.  Note
                                that $\rank{\alpha_2}, \rank{\alpha_3} \leq n$ since both $\gamma,\delta \in R_{n-1}$.
                                Moreover, from the induction hypothesis we know
                                that $R_0,\ldots,R_{n-1}$ are unambiguous. And
                                so, since $z \in L(\varphi_r) \cap
                                L(\overline{\varphi_r})$, we can use
                                Lemma~\ref{lem-meshset-disjoint} to conclude
                                that $\varphi_r = \overline{\varphi_r}$.
                        \item If $\alpha_3 = R_k$, then necessarily there exist such
                                productions $\eta,\overline{\eta} \in R_k$ that
                                $\varphi_l \varphi_r \in
                                \RewritingSet{\alpha_2}{\eta}$ whereas
                                $\overline{\varphi_l} \overline{\varphi_r} \in \RewritingSet{\alpha_2}{\overline{\eta}}$. Due to
                                Proposition~\ref{prop-rewritingset-rewrites}, we
                                know that $L(\varphi_l \varphi_r) \subseteq
                                L(\eta)$ and $L(\overline{\varphi_l}
                                \overline{\varphi_r}) \subseteq L(\overline{\eta})$. It
                                implies that $w z \in L(\eta) \cap L(\overline{\eta})$,
                                however, since $k < n$, we know from the induction
                                hypothesis that $R_k$ is unambiguous. Hence
                                $\eta = \overline{\eta}$. Finally, it means that
                                we can reduce this case to one of the previous cases when
                                $\alpha_3$ is complex, concluding that
                                $\varphi_l \varphi_r =
                                \overline{\varphi_l}\overline{\varphi_r}$.
                    \end{enumerate}
    \end{enumerate}
\end{proof}

\subsection{Generating functions}\label{sec:gen-funs}
Fix an arbitrary normal-order reduction grammar $R_n$. Let us consider the counting sequence $\{r_{n,k}\}_{k\in \mathbb{N}}$ where $r_{n,k}$ denotes the number of $S K$-combinators of size $k$ reducing in $n$ normal-order reduction steps. Suppose we associate with it a formal power series $R_n(z)$ defined as
\[R_n(z) = \sum_{k=0}^{\infty} r_{n,k}\, z^k. \]

In the following theorem we present a recursive method of computing the closed-form solution of $R_n(z)$ using the regular tree grammars $R_0,\ldots,R_n$ and the inductive use of the \emph{Symbolic Method} developed by Flajolet and Sedgewick~\cite{Flajolet:2009:AC:1506267}.\Needspace{3\baselineskip}

\begin{theorem}
For each $n \geq 0$, the ordinary generating function $R_n(z)$ corresponding to the sequence $\{r_{n,k}\}_{k\in \mathbb{N}}$ has a computable closed form solution.
\end{theorem}

\begin{proof}
Induction over $n$.  Let us start with giving previously computed closed-form solutions for $C(z)$, i.e.~the generating function corresponding to the set of all $S K$-combinators, and $R_0(z)$~\cite{Bendkowski2015}:
\begin{equation}\label{eq:r0(z)}
C(z) = \frac{1 - \sqrt{1 - 8z}}{2 z} \qquad \quad R_0(z) = \frac{1 - 2 z - \sqrt{1- 4 z - 4 z^2}}{2 z^2}.
\end{equation}
Clearly, both $C(z)$ and $R_0(z)$ are computable.

Now, suppose that $n \geq 1$. Recall that in its construction, $R_n$ might depend on previous reduction grammars $R_0, \ldots, R_{n-1}$, the set $\mathcal{C}$ of all $S K$-combinators and itself, via self-referencing productions. Due to Theorem~\ref{th-unambiguity}, $R_n$ is unambiguous and so we can express its generating function $R_n(z)$ as the unique solution of
\begin{equation}\label{eq:rn(z)-fun-eq}
R_n(z) = \sum_{\alpha \in R_n} z^{k(\alpha)} {C(z)}^{c(\alpha)} \prod_{i=0}^{n} {R_i(z)}^{r_i(\alpha)},
\end{equation}
where $k(\alpha)$, $c(\alpha)$ and $r_i(\alpha)$ denote respectively, the number of applications, the number of non-terminal symbols $C$ and the number of non-terminal symbols $R_i$ in $\alpha$.

Note that $R_n$ has exactly four self-referencing productions, i.e.~$S R_n$, $K R_n$, $S R_0 R_n$ and $S R_n R_0$. It means that by converting them into appropriate functional equations, we can further rewrite~(\ref{eq:rn(z)-fun-eq}) as
\begin{equation}\label{eq:rn(z)-fun-eq-2}
R_n(z) = 2 z R_n(z) + 2 z^2 R_0(z) R_n(z) + \sum_{\alpha \in \Phi(R_n)} z^{k(\alpha)} {C(z)}^{c(\alpha)} \prod_{i=0}^{n-1} {R_i(z)}^{r_i(\alpha)},
\end{equation}
where $\Phi(R_n)$ denotes the set of productions $\alpha \in R_n$ which do not reference $R_n$. By the induction hypothesis, we can compute the closed-form solutions for $R_0(z), \ldots, R_{n-1}(z)$ turning~(\ref{eq:rn(z)-fun-eq-2}) into a linear equation in $R_n(z)$. Simplifying~(\ref{eq:r0(z)}) for $R_0(z)$, we derive the final closed-form solution
\begin{equation*}\label{eq:rn(z)-fun-eq-simplified}
R_n(z) = \frac{1}{\sqrt{1-4 z - 4 z^2}} \sum_{\alpha \in \Phi(R_n)} z^{k(\alpha)} {C(z)}^{c(\alpha)} \prod_{i=0}^{n-1} {R_i(z)}^{r_i(\alpha)}.
\end{equation*}
\end{proof}

\subsection{Other applications}\label{sec:applications}
In this section we highlight some interesting consequences of the existence of normal-order reduction grammars. In particular, we prove that terms reducing in $n$ steps have necessarily bounded length.  Moreover, we show that the problem of deciding whether a given term reduces in $n$ steps, can be done in memory independent of the size of the term.

\begin{prop}\label{prop-LRn-production-length}
If $\alpha \in R_n$, then $\alpha$ has length at most $2n + 2$.
\end{prop}

\begin{proof}
    Induction over $n$. The base case $n = 0$ is clear from the shape of $R_0$.
    Fix $n > 0$. Let us consider long productions in $R_n$. If $\beta$ is a
    \textsc{K-Expansion} of some $X \alpha_1 \ldots \alpha_m \in R_{n-1}$, then
    \[ \beta = K(X \alpha_1 \ldots \alpha_k) \mathcal{C} \alpha_{k+1} \ldots
    \alpha_m \qquad \text{for} \quad 0 \leq k \leq m-1. \]
   
    Since setting $k = 0$ maximizes the length of $\beta$, we note that $\beta$
    is of length $m + 2$ and so by the induction hypothesis at most $2n + 2$.
    Now, let us consider the case when $\beta$ is a \textsc{S-Expansion} of some
    $X \alpha_1 \ldots \alpha_m \in R_{n-1}$. Then,
    \[ \beta = S(X \alpha_1 \ldots \alpha_k) \varphi_l \varphi_r \alpha_{k+3}
    \ldots \alpha_m \qquad \text{for} \quad 0 \leq k \leq m -2. \]
    where in addition $(\varphi_l\, \varphi_r) \in \RewritingSet{\alpha_{k+1}}{\alpha_{k+2}}$.
    Again, setting $k = 0$ maximizes the length of $\beta$. It follows that
    $\beta$ is of length at most $m + 1$ and so also at most $2n + 1$.
\end{proof}

In other words, terms reducing in $n$ steps cannot be too long as their length is tightly bounded by $2n + 2$. Now, let us consider the following two problems.

\begin{figure}[H]
\flushleft
\noindent \textbf{Problem:} \textsc{n-step-reducible}\\
\textbf{Input:} A combinatory logic term $x \in L(\mathcal{C})$.\\
\textbf{Output:} \textsc{yes} if and only if $x$ reduces in $n$ steps.
\end{figure}
\begin{figure}[H]
\flushleft
\noindent \textbf{Problem:} \textsc{reduces-in-n-steps}\\
\textbf{Input:} A combinatory logic term $x \in L(\mathcal{C})$ and a number $n
        \in \mathbb{N}$.\\
\textbf{Output:} \textsc{yes} if and only if $x$ reduces in $n$ steps.
\end{figure}

Since $n$ in not a part of the input, we can compute $R_n$ in constant time and memory. Using $R_n$ we build a bottom-up tree automaton recognizing $L(R_n)$~\cite{tata2007} and use it to check whether $x \in L(R_n)$ in time
$O(|x|)$, without using additional memory. On the other hand, the \textsc{Naive} algorithm requires $O(|x|)$ time and additional memory. At each reduction step, the considered term doubles at most in size, as $S x y z \to_w x z (y z)$. In order to find the next redex we spend up to linear time in the current size of $x$, therefore both size and time are bounded by

\begin{eqnarray*}
|x| + 2 |x| + 4 |x| + \cdots + 2^n |x| &=& |x| \Big(1 + 2 + 4 + \cdots + 2^n \Big) \\
 &=& |x| \Big( 2^{n+1} - 1 \Big) = O(|x|).
\end{eqnarray*}

As a natural extension, we get the following corollary.
\begin{cor}
        The \textsc{reduces-in-n-steps} problem is decidable in space
        depending exclusively on $n$, independently of $|x|$.
\end{cor}

\subsection{Upper bound}\label{sec:upper-bound}
In this section we focus on the upper bound on the number of productions in $R_n$. We show that there exists a primitive recursive function $f : \mathbb{N} \to \mathbb{N}$ such that $|R_n| \leq f(n)$.

Following the scheme of the soundness proofs in Section~\ref{sec:soundness}, we construct suitable upper bounds using the notions of tree potential and degree. In the end of this section, we show that these values are in fact bounded in each $R_n$, thus giving the desired upper bound.

\begin{lemma}\label{lem-meshset-upperbound}
    Let $\alpha,\beta$ be two trees of degree at most $n$ such that their total potential $\potential{\alpha} + \potential{\beta}$ is equal to $p$. Then, the number of distinct meshes in $\MeshSet{\alpha}{\beta}$ is bounded by ${|R_n|}^{e\, p!}$.
\end{lemma}

\begin{proof}
        Induction over the total potential $p$. Consider the following primitive recursive function $f_n :\mathbb{N} \to \mathbb{N}$.
    
    \begin{equation*}
  f_n(k) = \left\{\begin{array}{r@{}l@{\qquad}l}
    & 1 & \text{if }\ k = 0, \\[\jot]
    & {\left(|R_n| \cdot f_n(k-1)\right)}^{k} & \text{otherwise.}
  \end{array}\right.
\end{equation*}

    We claim that $|\MeshSet{\alpha}{\beta}| \leq f_n(p)$. Note that it suffices
    to consider such $\alpha,\beta$ that $|\MeshSet{\alpha}{\beta}| > 1$ since
    $f_n$ is an increasing function attaining positive values for any given input.
	It follows that the base case $p = 0$ is clear, as if $\potential{\alpha} + \potential{\beta} = 0$, then $\MeshSet{\alpha}{\beta}$ is necessarily empty.
	Now, let us assume that $p > 0$. From the construction of the common mesh set $M$ of $\alpha$ and $\beta$, we can distinguish two cases left to consider.
    \begin{enumerate}[(i)]
        \item Suppose that $\alpha = X \alpha_1 \ldots \alpha_m$ and
                $\beta = X \beta_1 \ldots \beta_m$. In order to maximize the
                size of $M$, we can furthermore assume that none of the pairs
                 $\alpha_i,\beta_i$ are rewritable. And so, the total
                number of meshes in $M$ is equal to the product of all meshes in
                corresponding mesh sets for $\alpha_i$ and $\beta_i$. The degree
                of $\alpha_i$ and $\beta_i$ is still at most $n$,
                however $\potential{\alpha_i} + \potential{\beta_i} \leq p-2$. Hence,
                using the induction hypothesis we get
                $ |\MeshSet{\alpha_i}{\beta_i}| \leq f_n(p-2)$.
                Since both $\alpha,\beta$ are of length $m \leq p$ we can furthermore
                state that
                \begin{eqnarray*}
                        |M| &\leq& \left(f_n(p-2)\right)^m \leq \left(f_n(p-2)\right)^p\\
                            &\leq& \left(f_n(p-1)\right)^p \leq 
                            \left(|R_n| \cdot f_n(p-1)\right)^p\\
                            &=& f_n(p).
                \end{eqnarray*}
        \item Let us assume w.l.o.g.~that $\alpha = R_i$ and $\beta$ is complex.
                In order to maximize the total number of meshes in $M$, we can
                moreover assume that all productions $\gamma \in R_i$ are
                similar to $\beta$ and generate disjoint sets of meshes. We claim
                 that $\MeshSet{\gamma}{\beta} \leq f_n(p-1)$. 
                 Clearly,
                 if $\gamma$ does not reference $R_i$, then our claim is trivially true.
                Suppose that $\gamma$ is a self-referencing production. If 
                $\gamma = X R_i$, then $\beta$ is in form of $X \beta_1$.
                From the construction of $M$, we get that \[|\MeshSet{\gamma}{\beta}| =
                 |\MeshSet{R_i}{\beta_1}|.\] As $\potential{R_i} + \potential{\beta_1}
                  \leq p - 1$, we can apply the induction hypothesis to 
                  $\MeshSet{R_i}{\beta_1}$ and immediately obtain 
                  $|\MeshSet{\gamma}{\beta}| \leq f_n(p-1)$. Now, suppose w.l.o.g.~that
                $\gamma = S R_i R_0$ and hence $\beta = S \beta_1 \beta_2$.
                Again, from the construction of $M$ we know that
                \[|\MeshSet{\gamma}{\beta}| = |\MeshSet{R_i}{\beta_1}| \cdot
                |\MeshSet{R_0}{\beta_2}|.\]
                Due to the fact that both $\potential{R_i} + \potential{\beta_1}
                 \leq p - 2$ and $\potential{R_0} + \potential{\beta_2} \leq p - 2$,
                  we can use the induction hypothesis and immediately get that
                \begin{eqnarray*}
                        |\MeshSet{\gamma}{\beta}| &=& |\MeshSet{R_i}{\beta_1}| \cdot
                |\MeshSet{R_0}{\beta_2}|\\
                        &\leq& f_n(p-2)\, f_n(p-2).
                \end{eqnarray*}
                Note that ${(f_n(p-2))}^2 \leq f_n(p-1)$ for $p \geq 2$ and, in
                 consequence, $|\MeshSet{\gamma}{\beta}| \leq f_n(p-1)$. Indeed, 
                 if $p = 2$, then
                 ${(f_n(p-2))}^2 = 1 \leq f_n(1) = |R_n|$. Otherwise if 
                 $p > 2$, then
                \begin{eqnarray*}
                        f_n(p-1) &=& {\left(|R_n| \cdot f_n(p-2)\right)}^{p-1}\\
                                 &=& {\left( {|R_n|}^{p-1} {(f_n(p-3))}^{p-2} \right)}^{p-1}\\
                                 &\geq& {\left( {|R_n|}^{p-2} {(f_n(p-3))}^{p-2} \right)}^{p-1}\\
                                 &=& {\left( |R_n| \cdot
        f_n(p-3)\right)}^{(p-1)(p-2)}.
                \end{eqnarray*}
                As
                $2 (p-2) \leq (p-1)(p-2)$ for $p > 2$, we finally obtain
				\begin{eqnarray*}
				{\left( |R_n| \cdot f_n(p-3)\right)}^{(p-1)(p-2)} &\geq& {\left(|R_n| \cdot f_n(p-3)\right)}^{2 (p-2)}\\
				&=& {(f_n(p-2))}^2.
				\end{eqnarray*}                
                
				We know therefore that $\MeshSet{\gamma}{\beta} \leq f_n(p-1)$ for 
				each $\gamma \in R_i$. Finally, using the fact that 
				$|R_i| \leq |R_n|$, we get
				\begin{eqnarray*}
                        |M| &\leq& |R_n| \cdot f_n(p-1)\\
                            &\leq& \left(|R_n| \cdot f_n(p-1)\right)^p\\
                            &=& f_n(p).
                \end{eqnarray*}
    \end{enumerate}
    And so, we know that $|\MeshSet{\alpha}{\beta}| \leq f_n(p)$. Solving the
    recurrence for $f_n(p)$, using e.g.~Mathematica \textregistered~\cite{mathematicaSoft}, we obtain the following closed form expression
    \[f_n(p) = {|R_n|}^{e\, p\, \gammaF{p}}, \]
    where
    \[ \Gamma(s,x) = (s - 1)!\, e^{-x} \sum_{k=0}^{s-1} \frac{x^k}{k!} \]
    is the upper incomplete gamma function (see e.g.~\cite{abramowitz-stegun1974}).
    Simplifying the above expression in the case $x = 1$ and using the observation
    that $\sum_{k=0}^{s-1} \frac{1}{k!} \leq e$ for arbitrary $s$, we finally
    obtain the anticipated upper bound
    \[ f_n(p) \leq {|R_n|}^{e\, p!}. \]
\end{proof}

\begin{lemma}\label{lem-rewritingset-size-upperbound}
    Let $\alpha,\beta$ be two trees of degree at most $n$ such that their total potential $\potential{\alpha} + \potential{\beta}$ is equal to $p$. Then, the number of distinct trees in $\RewritingSet{\alpha}{\beta}$ is bounded by ${|R_n|}^{1 + e\, p!}$.
\end{lemma}

\begin{proof}
	If $|\RewritingSet{\alpha}{\beta}| \leq 1$, then our claim is trivially true. Let us focus therefore on the remaining cases when either $\beta = X \beta_1 \ldots \beta_m$ and both $\beta_m$ and $\alpha$ are non-rewritable, or $\beta = R_i$. 
	
	First, consider the former case. Note that the resulting rewriting set is of equal size as
    $\MeshSet{\alpha}{\beta_m}$. Since $\potential{\alpha} + \potential{\beta_m}
    \leq p - 1$, we can use Lemma~\ref{lem-meshset-upperbound} to deduce that
    \[ |\RewritingSet{\alpha}{\beta}| = |\MeshSet{\alpha}{\beta_m}| \leq  {|R_n|}^{e\, (p-1)!} < {|R_n|}^{1 + e\, p!}.\]
	
	Now, let us consider the latter case. In order to maximize
    the resulting rewriting set we assume that each production $\gamma \in
    R_i$ generates a disjoint set of trees. We claim that each production
     $\gamma$ contributes at most ${|R_n|}^{e\, p!}$ new trees to the resulting
      rewriting set and therefore $|\RewritingSet{\alpha}{\beta}| \leq {|R_n|}^{1 + e\, p!}$, as there are at most $|R_n|$ productions in $R_i$. Indeed, consider an arbitrary $\gamma \in R_i$. Evidently, if $|\RewritingSet{\alpha}{\gamma}| \leq 1$, then our claim is true. Hence, let us assume that $|\RewritingSet{\alpha}{\gamma}| > 1$. It follows that $\gamma$ is complex. Let us rewrite it as $X \gamma_1 \ldots \gamma_m$. Note that as in the previous case, the resulting rewriting set is of equal size as $\MeshSet{\alpha}{\gamma_m}$. Since $\potential{\alpha} + \potential{\gamma_m} \leq p - 1$ we use Lemma~\ref{lem-meshset-upperbound} and get
      \[ |\RewritingSet{\alpha}{\gamma}| = |\MeshSet{\alpha}{\gamma_m}| \leq  {|R_n|}^{e\, (p-1)!} < {|R_n|}^{e\, p!}.\]
\end{proof}

\begin{lemma}\label{lem-meshset-potential-upperbound}
    Let $\alpha,\beta$ be two trees of total potential
    $\potential{\alpha} + \potential{\beta}$ equal to $p$. Then, each mesh in
    $\MeshSet{\alpha}{\beta}$ has potential bounded by $p! (1 + e)$.
\end{lemma}

\begin{proof}
        Induction over total potential $p$. Again, it suffices to consider such $\alpha, \beta$ that $\MeshSet{\alpha}{\beta}$ is not empty. Immediately, the base case $p = 0$ is clear. Let us assume that $p > 0$. Consider the following primitive recursive function $f :\mathbb{N} \to \mathbb{N}$.
\begin{equation*}
  f(k) = \left\{\begin{array}{r@{}l@{\qquad}l}
    & 1 & \text{if }\ k = 0, \\[\jot]
    & k \cdot \left( f(k-1) + 1 \right) & \text{otherwise.}
  \end{array}\right.
\end{equation*} 
    
	Let $\gamma \in \MeshSet{\alpha}{\beta}$. We claim that $\potential{\gamma}
    \leq f(p)$. Note that
    $f$ is an increasing function attaining positive values for any input. 
    We have two cases to consider.
    \begin{enumerate}[(i)]
        \item Suppose that $\alpha = X \alpha_1 \ldots \alpha_m$
                and $\beta = X \beta_1 \ldots \beta_m$. Note that
                $\potential{\alpha_i} + \potential{\beta_i} \leq p - 2$ 
                for each pair of corresponding arguments $\alpha_i, \beta_i$. Using
                 the induction hypothesis to pairs $\alpha_i, \beta_i$ and 
                 the fact that $\gamma \in \MeshSet{\alpha}{\beta}$ is 
                 similar to both $\alpha$ and $\beta$, we bound $\gamma$'s potential by
                \[\potential{\gamma} \leq m \cdot f(p-2) + m \leq p \cdot (f(p-2) + 1) \leq f(p). \]
        \item Assume w.l.o.g.~that $\alpha = R_i$ and $\beta$ is complex.
                It follows that $\gamma \in \MeshSet{\delta}{\beta}$ for some
                $\delta \in R_i$. If $\delta$ does not reference $R_i$, then
                clearly $\potential{\delta} \leq \potential{R_i} - 1$ and
                therefore $\potential{\gamma} \leq f(p-1)$. Now, suppose that
                $\delta$ is a self-referencing production of $R_i$. 
                
                If $\delta = X R_i$, then $\beta$ is in form of $X \beta_1$ and
                similarly $\gamma = X \gamma_1$. It follows that $\potential{\delta}
                 = \potential{R_i} + 1$ and therefore $\potential{\delta} +
                  \potential{\beta} = p + 1$. Note however that $\potential{\gamma_1}
                 \leq f(p-1)$ as $\potential{R_i} + \potential{\beta_1} \leq p - 1$.
                 Due to that, $\potential{\gamma} = 1 + f(p - 1) \leq f(p)$.
                 
                 Let us assume w.l.o.g.~that $\delta = S R_i R_0$. Immediately,
                 $\beta$ is in form of $S \beta_1 \beta_2$ whereas $\gamma = S \gamma_1 	
                 \gamma_2$. Moreover, $\potential{\delta} = \potential{R_i} + 3$. Note 
                 however that both $\potential{R_i} + \potential{\beta_1} \leq p - 2$
                 and $\potential{R_0} + \potential{\beta_2} \leq p - 2$. We can
                  therefore use the induction hypothesis and conclude that
                  \[\potential{\gamma} = 2 + \potential{\gamma_1} + \potential{\gamma_2}
                  \leq 2 + 2 \cdot f(p-2).\]
				Since $\potential{\delta} \geq 4$, we know that $p \geq 3$ and so we can further bound $\potential{\gamma}$ by
				\begin{eqnarray*}
				\potential{\gamma} &= 2 \left(1 + f(p-2)\right)\\
				&\leq (p-1) \left( 1 + f(p-2) \right)\\
				&=f(p-1) \leq f(p).
				\end{eqnarray*}
    \end{enumerate}
    Finally, we know that $\potential{\gamma} \leq f(p)$. What remains is to solve
    the recursion, using e.g.~Mathematica \textregistered~\cite{mathematicaSoft}, for $f$ and give its closed form solution. It follows that
    \begin{eqnarray*}
        f(p) &= \Gamma (1 + p) + e\, p\, \gammaF{p}\\
             &\leq p! + e\, p!\\
             &= p! (1 + e)
    \end{eqnarray*}
    where \[ \Gamma(n) = (n - 1)! \]
\end{proof}

\begin{lemma}\label{lem-rewritingset-potential-upperbound}
    Let $\alpha,\beta$ be two trees of potential
    $\potential{\alpha} + \potential{\beta} = p$. Then, each tree in
    $\RewritingSet{\alpha}{\beta}$ has potential bounded by $p! (1 + e) + p$.
\end{lemma}

\begin{proof}
    Let $\gamma$ be an arbitrary tree in $\RewritingSet{\alpha}{\beta}$. Based
    on the structure of $\beta$ we have several cases to consider.    
    If $\beta = \mathcal{C}$, then $\gamma = \mathcal{C} \alpha$ and so
     $\potential{\gamma} = \potential{\alpha} + 1 = p + 1$. Note
    that $1 < p! (1 + e)$ for any $p$ and thus our bound holds. 
    
    If $\beta = X
    \beta_1 \ldots \beta_m$, then $\potential{\alpha} + \potential{\beta_m} \leq p - 1$. 	In both cases when 	$\alpha \bowtie \beta_m$ the resulting tree has potential bounded
     by $p$ and so also by $p! (1 + e) + p$. Let us assume that $\alpha \parallel
      \beta_m$. We can therefore rewrite $\gamma$ as $X \gamma_1 \ldots \gamma_m$.
    Using Lemma~\ref{lem-meshset-potential-upperbound}, we know that
    $\potential{\gamma_m} \leq (p-1)! (1 + e)$. Moreover, both $\alpha$ and $\beta$
     are similar to $\gamma$. Let us rewrite them as $X \alpha_1 \ldots \alpha_m$ and $X \beta_1, \ldots, \beta_m$, respectively. Note that for each  $i < m$, $\gamma_i$ is  equal to $\alpha_i$ or $\beta_i$. It follows that we can bound the potential of $X \gamma_1 \ldots \gamma_{m-1}$ by $p-1$ and hence $\gamma$'s potential by $(p-1)! (1 + e) + p$.

    Now, if $\beta = R_i$, then $\gamma \in \RewritingSet{\alpha}{\delta}$
    for some $\delta \in R_i$. Clearly, if $\delta$ does not reference $R_i$, we know that $\potential{\delta} \leq \potential{R_i} - 1 \leq p - 1$. Moreover, $\delta$ is complex, as otherwise $\RewritingSet{\alpha}{\delta} = \emptyset$. Using our previous argumentation, we can therefore conclude that $\potential{\gamma} \leq (p-1)! (1 + e) + p$.
    Suppose that $\delta$ is a self-referencing production of $R_i$. If $\delta = X R_i$, then $\alpha$ is in form of $X \alpha_1$ and $\gamma = X \gamma_1$. Immediately,
    $\potential{\alpha} + \potential{\delta} = p + 1$. If $R_i \bowtie \alpha_1$, then $\gamma$ has potential bounded by $p$. Therefore, let us assume
    that $R_i \parallel \alpha_1$. Since $\potential{R_i} + \potential{\alpha_1} = p - 1$, we know from Lemma~\ref{lem-meshset-potential-upperbound} that $\potential{\gamma_1} \leq (p-1)! (1 + e)$. It follows immediately that $\potential{\gamma} \leq (p-1)! (1 + e) + 1 \leq p! (1 + e) + p$. 
    
    Finally, suppose that $\delta = S \delta_1 \delta_2$ and so
	$\alpha = S \alpha_1 \alpha_2$. Immediately, $\gamma = S \gamma_1 \gamma_2$.
	 Again, if $\delta_2 \bowtie \alpha_2$,
	 we can bound $\gamma$'s potential by $p$. Hence, let us assume that $\delta_2 \parallel \alpha_2$. Clearly, $\potential{\alpha} + \potential{\delta} = p + 3$.
	 Note however that $\potential{\alpha_1} + \potential{\delta_1} \leq p - 2$
	 and $\potential{\alpha_2} + \potential{\delta_2} \leq p - 2$, as
	 both $\delta_1$ and $\delta_2$ are non-terminal reduction grammar symbols of positive potential. Using Lemma~\ref{lem-meshset-potential-upperbound} to
	 $\MeshSet{\alpha_2}{\delta_2}$ we conclude that $\potential{\gamma_2} \leq (p-2)! (1 + e)$. It follows that $\potential{\gamma} \leq (p-2)! (1 + e) + p \leq p! (1 + e) + p$.
\end{proof}

\begin{lemma}\label{prop-Rn-potential-upperbound}
    There exists a primitive recursive function $\psi : \mathbb{N} \to \mathbb{N}$ such that $\potential{R_n} \leq \psi(n)$.
\end{lemma}

\begin{proof}
	Consider the following function $\psi : \mathbb{N} \to \mathbb{N}$:
	\begin{equation*}
  \psi(k) = \left\{\begin{array}{r@{}l@{\qquad}l}
    & 1 & \text{if }\ k = 0, \\[\jot]
    & 4 \left( \psi(k-1) + 2 \right)! + 2 \psi(k-1) + 5 & \text{otherwise.}
  \end{array}\right.
\end{equation*}

    Clearly, $\psi$ is an increasing primitive recursive function.
    We show that $\psi(n)$ bounds the potential of $R_n$ using induction over $n$.
    Since $\potential{R_0} = \psi(0) = 1$, the base case is clear.
    Let $n > 0$. In order to prove our claim, we have to check that $\potential{\alpha} \leq \psi(n) - 1$ for all productions $\alpha \in R_n$ which do not reference $R_n$.
     \begin{enumerate}[(i)]
     \item Suppose that $\alpha = S R_{n-i} R_i$. Clearly, the potential of $\alpha$ is equal to $2 + \potential{R_{n-i}} + \potential{R_i}$. Using the induction hypothesis, we know moreover that
     \begin{eqnarray*}
     \potential{\alpha} &\leq& 2 + \psi(n-i) + \psi(i)\\
     &\leq& 2 + 2 \psi(n-1)\\
     &\leq& \psi(n) - 1.
     \end{eqnarray*}
     \item Let $\alpha = K R_{n-1} \mathcal{C}$. Due to the fact that $\potential{\alpha} = 2 + \potential{R_{n-1}}$, we use the induction hypothesis and immediately obtain \[\potential{\alpha} \leq 2 + \psi(n-1) \leq \psi(n) - 1.\]
     \item Suppose that $\alpha \in \KExpansions{\beta}$ for some $\beta \in R_{n-1}$. Note that $\potential{\beta} \leq \psi(n-1) + 3$ as the productions of greatest potential in $R_{n-1}$ are exactly $S R_{n-1} R_0$ and $S R_0 R_{n-1}$. Since $\potential{\alpha} = 2 + \potential{\beta}$, we get
     \[ \potential{\alpha} \leq 5 + \psi(n-1) \leq \psi(n) - 1. \]
     \item Finally, let $\alpha \in \SExpansions{\beta}$ for some $\beta \in R_{n-1}$. Again, $\potential{\beta} \leq \potential{R_{n-1}} + 3$ and hence from the induction hypothesis $\potential{\beta} \leq \psi(n-1) + 3$. Let us rewrite $\alpha$ as $S(X \beta_1 \ldots \beta_k) \varphi_l \varphi_r \beta_{k+3} \ldots \beta_m$ where
     $\beta = X \beta_1 \ldots \beta_m$. Note that
     $\potential{\alpha} \leq \potential{\beta} + \potential{\varphi_l} + \potential{\varphi_r} + 1$. Moreover, as $\potential{\varphi_l \varphi_r} = 1 + \potential{\varphi_l} + \potential{\varphi_r}$, we get
     $ \potential{\alpha} \leq \potential{\beta} + \potential{\varphi_l \varphi_r}$.
     Since $\potential{\beta_{k+1} \beta_{k+2}} \leq \potential{\beta} - 1$ and thus, $\potential{\beta_{k+1} \beta_{k+2}} \leq \psi(n-1) + 2$, we can use Lemma~\ref{lem-rewritingset-potential-upperbound}
     to obtain \[\potential{\varphi_l \varphi_r} \leq
     (\psi(n-1) + 2)! (1 + e) + \psi(n-1) + 2.\]
     It follows therefore that
     \begin{eqnarray*}
     \potential{\alpha} &\leq& \potential{\beta} + \potential{\varphi_l \varphi_r}\\
     &\leq& (\psi(n-1) + 2)! (1 + e) + 2 \psi(n-1) + 5\\
     &\leq& \psi(n) - 1
     \end{eqnarray*}
     where the last inequality follows from the fact that
     \[ \left( 3 - e \right) (\psi(n-1) + 2)! \geq \frac{1}{5} (\psi(n-1) + 2)! \geq \frac{6}{5} \geq 0. \]
     \end{enumerate}
\end{proof}

\begin{theorem}
 There exists a primitive recursive function $\chi : \mathbb{N} \to \mathbb{N}$
    such that the number $|R_n|$ of productions in $R_n$ is bounded by $\chi(n)$.
\end{theorem}

\begin{proof}
Consider $R_n$ for some $n > 0$. Note that $R_n$ consists of:
\begin{enumerate}[(i)]
\item two productions $S R_n$ and $K R_n$,
\item $n + 1$ short $S$-productions in form of $S R_{n-i} R_i$,
\item an additional $K$-production $K R_{n-1} \mathcal{C}$,
\item $\KExpansions{\alpha}$ for each $\alpha \in R_{n-1}$ and
\item $\SExpansions{\alpha}$ for each $\alpha \in R_{n-1}$.
\end{enumerate}

It suffices therefore to bound the number of \textsc{K-} and \textsc{S-Expansions}, as
the number of other productions in $R_n$ is clear. Let us start with \textsc{K-Expansions}. Suppose that $\alpha$ is of length $m$. Clearly, $|\KExpansions{\alpha}| = m$. Using Proposition~\ref{prop-LRn-production-length}, we know that that each production $\alpha \in R_{n-1}$ is of length at most $2n$. It follows that there are at most $2 n \cdot |R_{n-1}|$ \textsc{K-Expansions} in $R_n$. Now, let us consider \textsc{S-Expansions}. In order to bound the number of \textsc{S-Expansions} in $R_n$, we assume that each production $\alpha \in R_{n-1}$ is of length $2n$ and moreover each \textsc{RewritingSet} of appropriate portions of $\alpha$ generates a worst-case set of trees. And so, assuming that $\alpha$ is of length $2n$ we can rewrite it as $ X \alpha_1 \ldots \alpha_{2n}$. Let $\psi$ denote the upper bound function on the potential of $R_{n-1}$ from Lemma~\ref{prop-Rn-potential-upperbound}. Evidently, $\potential{\alpha} \leq \psi(n-1) + 3$. Now, using Lemma~\ref{lem-rewritingset-size-upperbound} we know that each $\RewritingSet{\alpha_i}{\alpha_{i+1}}$ contributes at most \[{|R_{n-1}|}^{1 + e\big(\psi(n-1) + 3\big)!} \]
new \textsc{S-Expansions}. As there are at most $2n-1$ pairs of indices $(i,i+1)$ yielding \textsc{RewritingSets}, we get that the number of \textsc{S-Expansions} in $R_n$ is bounded by
\[ (2n-1) \cdot |R_{n-1}| \cdot {|R_{n-1}|}^{1 + e\big(\psi(n-1) + 3\big)!} \leq (2n-1) \cdot {|R_{n-1}|}^{2 + 3 \big(\psi(n-1) + 3\big)!}. \]

Finally, since $|R_0| = 5$, we combine the above observations and get the following primitive recursive upper bound on $|R_n|$.
\begin{equation*}
  \chi(k) = \left\{\begin{array}{r@{}l@{\qquad}l}
    & 5 & \text{if }\ k = 0, \\[\jot]
    & 4 + k + 2k \cdot \chi(k-1) &\\
    & + (2k-1) \cdot {\chi(k-1)}^{2 + 3 \big(\psi(k-1) + 3\big)!} & \text{otherwise.}
  \end{array}\right.
\end{equation*}
\end{proof}

\section{Conclusion}\label{sec:conclusion}
We gave a complete syntactic characterization of normal-order reduction for combinatory logic over the set of primitive combinators $S$ and $K$. Our characterization uses regular tree grammars and therefore exhibits interesting
 algorithmic applications, including the computation of corresponding generating functions. We investigated the complexity of the generated reduction grammars, giving a primitive recursive upper bound on the number of their productions. We emphasize the fact that although the size of $R_n$ is bounded by a primitive recursive function of $n$, it seems to be enormously overestimated. Our computer implementation of the \textsc{Reduction Grammar} algorithm~\cite{mb-haskell-implementation} suggests that the first few numbers in the sequence ${\{|R_n|\}}_{n\in\mathbb{N}}$ are in fact
\[ 5, 12, 75, 625, 5673, 53164, 508199, \ldots \]

\vspace{2mm}
The upper bound $\chi(1)$ on the size of $R_1$ is already of order $6 \cdot 10^{84549}$, whereas the actual size of $R_1$ is equal to $12$. Naturally,
we conjecture that ${\{R_n\}}_{n\in\mathbb{N}}$ grows much slower than
${\{\chi(n)\}}_{n\in\mathbb{N}}$, although the intriguing problem of 
giving better approximations on the size of $R_n$ for large $n$
 is still open.
 
\section*{Acknowledgements}
We would like to thank Katarzyna Grygiel for many fruitful discussions and valuable comments.

\bibliographystyle{plain}
\bibliography{references}

\end{document}